\documentclass[10pt,a4paper ]{article}
\usepackage{fullpage}
\usepackage[T1]{fontenc}
\usepackage[tbtags]{amsmath}
\usepackage{amsfonts,amssymb,amsthm,euscript,pifont}
\usepackage{mathrsfs}
\usepackage{textcomp}
\usepackage{euscript,pifont,boxedminipage}
\usepackage{float}
\usepackage[colorlinks=true]{hyperref}
\newcommand{\Cc}{\mathcal C}

\newcommand{\Vv}{\mathcal V}

\newcommand{\setA}{\mathbb A}

\newcommand{\EE}{\mathbb E}
\newcommand{\NN}{\mathbb N}
\newcommand{\PP}{\mathbb P}
\newcommand{\RR}{\mathbb R}
\newcommand{\TT}{\mathbb T}
\newcommand{\ZZ}{\mathbb Z}

\newcommand{\er} {\mathbb R}
\newcommand{\Hh}{{\mathcal H}}
\newcommand{\tHh}{{\widetilde{\mathcal{H}}}}
\newcommand{\Hu}[1]{{\rm{($\text{H}_{\text{unif}({#1})}$)}}}
\newcommand{\Hw}{{\rm{($\text{H}_{\omega}$)}}}

\newcommand\1{\leavevmode\hbox{\rm \small1\kern-0.35em\normalsize1}}
\theoremstyle{plain}
\newtheorem{thm}{Theorem}[section]
\newtheorem{prop}[thm]{Proposition}

\newtheorem{remark}[thm]{Remark}
\newtheorem{lemma}[thm]{Lemma}

\newtheorem{corollary}[thm]{Corollary}

\numberwithin{equation}{section}
\title{Local  existence of  analytical solutions to an incompressible  Lagrangian stochastic model in a periodic domain}
\author{Mireille Bossy\thanks{INRIA Sophia Antipolis France, EPI TOSCA; Mireille.Bossy@inria.fr}
\and Joaquin Fontbona\thanks{DIM--CMM,  UMI 2807 UChile-CNRS, Universidad de Chile, Casilla 170-3, Correo 3, Santiago,
Chile ;  Partially supported by Fondecyt  Grant 1110923  and Basal-Conicyt-Chile; fontbona@dim.uchile.cl.}
\and Pierre-Emmanuel Jabin\thanks{CSCAMM--Department of Mathematics, University of Maryland, College Park; pjabin@cscamm.umd.edu}
\and Jean-Fran\c{c}ois Jabir\thanks{CIMFAV, Universidad de Valpara\'iso, Chile;
 Partially supported by Fondecyt Grant 3100132 and Conicyt  PAI/ACADEMIA $79090016$;  jean-francois.jabir@uv.cl.}}
\date\today
\begin{document}
\maketitle
\begin{abstract}
We consider an incompressible kinetic Fokker Planck equation in the flat torus, which is a simplified version of the Lagrangian stochastic models for turbulent flows introduced by S.B. Pope in the context of computational fluid dynamics.  The main  difficulties  in its treatment arise from a pressure type force that couples the Fokker Planck equation with a Poisson equation which strongly depends on the second order moments of the fluid velocity. In this paper we  prove short time existence of analytic solutions in the one-dimensional case, for which we are able to use techniques and functional norms that have been recently introduced    in the  study  of a related singular model.
\end{abstract}
\centerline{{\bf Keywords} Fluid particle model;  Incompressibility;   Singular kinetic equation;  Analytic solution.}
\smallskip

\centerline{{\bf AMS 2010 Subject Classification   } 35Q83,
35Q84  ,
82C31.}
\tableofcontents
\section{Introduction}

Let $\TT^{d}:=\er^d\slash \ZZ^{d}$  denote  the flat $d$-dimensional torus and  $\beta\geq 0$,  $\sigma \in \er$ and   $\alpha\in \{0,1\}$ be fixed constants.
 We consider the  following  partial differential equation  in $\TT^{d}\times\er^{d}$ with  (scalar)  unknown  functions  $f(t,x,u)$ and $P(t,x)$:
 \begin{subequations}\label{eq:CVFP}
\begin{align}
&\partial _{t}f(t,x,u) +  u\cdot \nabla_x f(t,x,u)  =
     \nabla_{u}f(t,x,u) \cdot  \left(\nabla_{x}P(t,x)+\beta \left( u-  \alpha \int_{\er^{d}}v f(t,x,v)\,dv \right)\right)   \nonumber  \\   & \hspace{4.8cm}  + \frac{\sigma^2}{2} \triangle_u f(t,x,u)  + \beta \, d \,  f(t,x,u)
\hspace{1cm}  \mbox{ on }(0,T]\times\TT^{d}\times\er^{d},    \label{eq:pde_Fokker-Planck} \\
&f(0,x,u)=f_{0}(x,u) \mbox{ on }\TT^{d}\times\er^{d} \mbox{ and }  \label{eq:pde_init_cond} \\
& \int_{\er^{d}}f(t,x,u)\,du =1\mbox{ on }[0,T]\times\TT^{d}.  \label{eq:pde_mass_constraint}
\end{align}
\end{subequations}
This ``constrained''  equation of  kinetic type can be understood as a the Fokker-Planck equation associated with the  stochastic differential equation   in $\TT^{d}\times\er^{d}$ :
  \begin{subequations}
\label{eq:generic_incompressible_lagrangian}
\begin{align}
&X_{t}=\left[ X_{0} +\int_{0}^{t} U_{s}\,ds\right]  , \quad  U_{t}=U_{0}+\sigma W_{t} -\int_{0}^{t}\nabla_{x}  P (s,X_{s})ds-
\beta\int_{0}^{t}(U_{s} -\alpha \EE(U_s |X_s))ds  \label{SDE}
\\
& law (X_0,U_0) =f_0(x,u)dx\, du  , \label{eq:sde_init_cond}
\\
&\PP(X_{t}\in dx)=dx, \mbox{ for all }t\in[0,T] ,  \label{eq:mass_constraint}
\end{align}
\end{subequations}
where  the drift term is unknown and where $x\mapsto [x]$ denotes the projection on the torus.    Equation \eqref{eq:generic_incompressible_lagrangian}  constitutes  a laboratory example of the class of Lagrangian stochastic models for incompressible turbulent flows, introduced mainly  by S.B. Pope in the eighties in order to  provide  a  fluid-particle description of turbulent flows and  develop probabilistic numerical methods  for their simulation.  We refer the reader to \cite{Pope-94}  for a general  presentation of this turbulent model approach in the framework of computational fluid dynamics, and to  \cite{jabir-10b},~\cite{jabir-10a} for a survey on mathematical problems raised by  the Lagrangian stochastic models.
 In physical terms, when $\alpha=0$, the process $U_t$ representing the velocity of a fluid particle reverts towards the origin like an Ornstein-Uhlenbeck  process with a  potential given by the standard kinetic energy $\EE |U_t|^2$. When $\alpha=1$,  reversion towards the origin in \eqref{SDE}  is replaced by  reversion towards the {\it averaged velocity} or {\it bulk-velocity}, $\EE(U_t |X_t=x) $,  which can be associated  to the local-in-space potential   $\EE(| U_t - \EE(U_t |X_t)|^2|X_t=x)$,  interpreted as the {\it turbulent} kinetic energy (notice  that under condition  \eqref{eq:pde_mass_constraint}  or \eqref{eq:mass_constraint}   $ \int_{\er^{d}} v f(t,x,v)dv$  is the conditional expectation $\EE(U_t |X_t=x)$). In both cases,
 the additional drift  term $\nabla_x P(t,x)$ is interpreted as the gradient  of a pressure field intended to accomplish the  homogeneous mass distribution  constraint specified by  equations  \eqref{eq:pde_mass_constraint} or \eqref{eq:mass_constraint},  in other words to force the particle position  $X_t$ to have a macroscopically uniform spacial distribution.

In spite of its relevance for the simulation of complex fluid dynamics (see e.g. \cite{minier-peirano-01}, \cite{Pope-03} and the references therein), a rigorous mathematical formulation of the Lagrangian stochastic models, and in particular of the  uniform mass distribution constraint, has not yet been given.  Indeed, equations \eqref{eq:CVFP}  and \eqref{eq:generic_incompressible_lagrangian} exhibit several conceptual and technical difficulties, and to our knowledge there is so far no direct strategy for its study or mathematical results  about it, neither in the field of
stochastic processes nor in that  of kinetic PDE. In \cite{jabir-10a}, first well-posedness results on a simpler kinetic model were obtained,  which featured nonlinearity of conditional type. From a probabilistic point of view, the conditional expectation was treated as  a McKean-Vlasov equation. This enabled the authors to also construct  a mean field stochastic particle approximation of the nonlinear model. Combined with  an heuristic  numerical procedure  to deal with the   constraint  \eqref{eq:mass_constraint}  and the pressure term, that particle scheme  gave rise to a stochastic numerical downscaling  method  studied and implemented in \cite{jabir-10b}.  Extensions of some of those results to a relevant instance  of   boundary
  value problem were obtained in \cite{jabir-11a} and \cite{jabir-11c}.
However,  in spite of the formal resemblance  of the uniform mass distribution   with the case of the incompressible Navier-Stokes equations,
there is so far no rigorous mathematical evidence that   \eqref{eq:mass_constraint}  can be satisfied  by adding  a force term of the form $\nabla_{x}P(t,x)$ in the linear  Langevin process (a trivial exception is the situation $\nabla_{x} P\equiv 0$  of the  stationary Langevin process, considered as a benchmark for the stochastic downscaling method in~\cite{jabir-10b}).

The aim of this paper is to address  for the first time the well-posedness of a relatively simple instance of  Lagrangian stochastic models, yet satisfying in a non trivial way the  uniform mass distribution constraint.  A first step in our study will be to  establish  an alternative  formulation of the previous equation.
In the Lagrangian modeling of turbulent flow, the constraint \eqref{eq:mass_constraint} is indeed formulated  heuristically by rather imposing some  divergence free property on the flow, which in the case of  system  \eqref{eq:generic_incompressible_lagrangian} would correspond to a divergence free condition on the bulk velocity field:
\begin{equation*}
\nabla_x  \cdot  \EE(U_t |X_t=x) =0.
\end{equation*}
By taking the divergence of a formal equation for the  bulk velocity derived from the Fokker-Planck equation,  and resorting to a classical projection argument on the space of divergence free fields, it is then assumed that the field $P$  verifies an elliptic PDE, which in our notation  is written as
\begin{equation}\label{ellipticPDE}
\triangle_x P(t,x) = -\sum_{i,j=1}^d \partial _{x_i x_j} \EE\left(U^{(i)}_{t}U^{(j)}_{t}|X_t=x\right)
\end{equation}
(see \cite{Pope-03} for a precise formulation and related numerical issues).
 Consistently with this heuristic point of view,  we will rigorously show below that, under natural assumptions on the initial data, any smooth  pair  $(f,P)$  that is a classical solution to \eqref{eq:CVFP}  must also be a solution to the system
\begin{equation}\label{eq:VFPalternative}
\left\{
\begin{aligned}
&\partial _{t}f(t,x,u) +  u\cdot \nabla_x f(t,x,u)  =
   \frac{\sigma^2}{2} \triangle_u f(t,x,u)  + \beta \, d \,  f(t,x,u)  + \beta u\cdot \nabla_u f(t,x,u)   \\
& \qquad +  \nabla_{u}f(t,x,u) \cdot  \left(\nabla_{x}P(t,x)  -\beta \alpha \int_{\er^{d}}v f(t,x,v)\,dv \right)=0\mbox{ on }(0,T]\times\TT^{d}\times\er^{d},   \\
 &
 f(0,x,u)=f_{0}(x,u)\mbox{ on }\TT^{d}\times\er^{d}, \\
&  \triangle_x P(t,x)=-\sum_{i,j=1}^{d}\partial_{x_i x_j} \int_{\er^d} v_i v_j  f(t,x,v)dv\mbox{ on }[0,T]\times\TT^{d},
\end{aligned}
\right.
\end{equation}
where, plainly, condition \eqref{eq:pde_mass_constraint} has been replaced by the above Poisson equation. The two systems  however seem not to be equivalent in general.  From the PDE  point of view, the  interest of formulation \eqref{eq:VFPalternative} is that it allows us to see the original problem as an instance of Vlasov-Fokker-Planck type equation, albeit  highly singular:  the gradient of the pressure field turns out to be the convolution of the derivative of the periodic Poisson kernel with the function $-\sum_{i,j=1}^{d}\partial_{x_i x_j} \int_{\er^d} v_i v_j  f(t,x,v)dv$.
 Existence of smooth solutions  to  nonlinear kinetic equations with singular potential has been addressed  in several situations, mainly recently,  see e.g.  \cite{Benachour-89}, \cite{Mouhot-Villani-11} for the  nonlinear Vlasov-Poisson equation and  \cite{Besse-al-09}, \cite{Besse-11}  and \cite{Ghendrih-al-09} for gyro-kinetic models. However,  the  simultaneous  regularity control of the solution and its second  moments,  required to rigorously formulate  equation \eqref{eq:VFPalternative},  does not fall into previous  mathematical frameworks.

In the case $d=1$, we can specify $P(t,x)$ on $[0,T]\times\er$ by $P(t,x) = - \int_\er u^2f(t,x,u)du$. Hence, in the present paper, we restrict ourselves to the simpler situation of the one-dimensional equation
\begin{equation}\label{eq:VFP_1D}
\left\{
\begin{aligned}
&\partial _{t}f(t,x,u) +  u\cdot \partial_x f(t,x,u)  =
\frac{\sigma^2}{2}\partial^2_u f(t,x,u)+\beta f(t,x,u)+\beta u\partial_u f(t,x,u)\\
& \qquad +  \partial_{u}f(t,x,u)\left(\partial_{x}P(t,x)  -\beta \alpha \int_{\er}v f(t,x,v)dv\right)=0\mbox{ on }(0,T]\times\TT\times\er,\\
&f(t,x,u)=f_{0}(x,u)\mbox{ on }\TT\times\er, \\
&P(t,x) = - \int_\er u^2f(t,x,u)du\mbox{ on }[0,T]\times\TT.
\end{aligned}
\right.
\end{equation}
 To tackle  the system \eqref{eq:VFP_1D}, we will follow new ideas introduced  in \cite{JabNou-11}, in order to obtain a local existence result of analytical solutions.
Our main results are valid irrespective of whether   $\sigma\neq 0$ or $\sigma=0$, and hold for any $\beta\in \RR$. We summarize them in the following  simplified  statement:
\begin{thm}\label{thm:mainmain} Let $\bar{\lambda}>0$ and $s\geq 4$ be an  even integer. There exist a constant
$\kappa_0 =\kappa_0(\bar{\lambda},s)$ and a positive function $r \mapsto \kappa_1(r,\bar{\lambda},s)$  such that if   $f_0:\TT\times \er \to \er$ of class ${\cal C}^{\infty}$ and $T>0$ satisfy:
\begin{itemize}
\item $\int_{\er}f_0(x,u)du=1$ and $\partial_x\int_{\er^d}u  f_0(x,u)du=0$ for all $x \in \TT$,
\item  $\Vert (1+u^2)^{\frac{s}{2}} \partial^{l}_x \partial^{k}_u f_0 \Vert_{\infty}\leq \frac{C_0(k+m)! (l +n)!}{\overline{\lambda}^{k+l  }}$ for some $n,m\in \NN$, all pair of indices $k,l \in \NN$ and a constant $C_0<\kappa_0(\bar{\lambda},s)$,  and
\item  $T<\kappa_1(C_0,\bar{\lambda},s)$,
\end{itemize}
then a  classic smooth solution  $f$  to equation \eqref{eq:VFP_1D} exists in  $[0,T]\times\TT\times\er$  and satisfies:
$\int_{\er}f(t,x,u)du=1$ and $\partial_x\int_{\er^d}u  f(t,x,u)du=0$ for all $(t,x) \in [0,T]\times \TT$ .
\end{thm}

In the case $d=1$ and $\sigma\neq 0$, this result will yield the following   statement on the `` incompressible Langevin process''  \eqref{eq:generic_incompressible_lagrangian}:

 \begin{corollary} Under the assumptions Theorem \ref {thm:mainmain},  there exists in $[0,T]$ a solution $(X_t,U_t)\in \TT\times \RR$  of the singular McKean-Vlasov SDE
  \begin{equation*}
\begin{split}
&X_{t}=\left[ X_{0} +\int_{0}^{t} U_{s}\,ds\right]  \\
&   U_{t}=U_{0}+\sigma W_{t} + \int_{0}^{t} \partial_x  \left[  \int_{\er} u^2 p_s(\cdot,u)du\right]     (X_{s})ds-
\beta\int_{0}^{t}\left(U_{s} -\alpha   \int_{\er} u \ p_s( X_{s},u)du \right)ds
\\
& law (X_t,U_t) =p_t(x,u)dx\, du ,  \quad p_0(x,u)=f_0(x,u).
\end{split}
\end{equation*}
Moreover,  $law(X_t)=dx$ for all $t\in [0,T]$ and  $(X,U)$ is a solution of  the stochastic differential equation  \eqref{eq:generic_incompressible_lagrangian} with the pressure field $P(t,x) = - \int_\er u^2p_t(x,u)du$.

 \end{corollary}

The remainder of the paper is organized as follows:

\medskip

In  Subsection \ref{sect:Poisson} we briefly establish the  validity of system \eqref{eq:VFPalternative} for any solution to equation \eqref{eq:CVFP} in arbitrary space dimension, and  state additional conditions required  in order that, reciprocally,  a solution to the former also solves the latter.
From Section  \ref{1d_beta=0}  on,  we restrict ourselves to the one-dimensional case. We  recall therein the analytical norms and seminorms introduced in \cite{JabNou-11} and we state useful properties of them. Following  their strategy,  in the case $\beta=0$ we then  introduce  an equivalent formulation of equation  \eqref{eq:VFP_1D}, in order to deal with the integrability problems posed by the first and second order velocity moments  involved in the equation. We then show that solutions to \eqref{eq:VFPalternative}  in these particular spaces of analytical functions  actually do satisfy the conditions  required to be solutions of \eqref{eq:CVFP}. Using the fixed point argument of  \cite{JabNou-11}, we will then  prove a local existence result in these analytical spaces, which indeed is a slightly more general  version of  Theorem \ref{thm:mainmain} restricted to the case $\beta=0$.
In Section  \ref{sec:Kinet1d}, we extend the previous result to the case $\beta\geq0$.  In Section \ref{sec:SDE} we deduce from the previous sections a local existence result for the stochastic differential equation \eqref{eq:generic_incompressible_lagrangian}. Finally, some technical results are proved in the Appendix section.

\bigskip

We fix some notation to be used throughout:
\begin{itemize}
\item $T>0$ is a fixed time horizon.
\item Functions $[0,T]\times \TT^d \times\RR^d \ni (t,x,v)  \mapsto \phi(t,x,v)  \in \RR$ are identified with functions
 $ [0,T]\times \RR^d \times\RR^d \ni (t,x,v)  \mapsto \phi(t,x,v)  \in \RR$ that are $1$- periodic in the variable $x$. Similar identification are made for functions defined on $[0,T]\times \TT^d$ and $\TT^d$.
\item Given $T>0$ and $d \in \NN$, a  function  $\phi: [0,T] \times E \to \RR $ with $E=\RR^d\times \RR^d, \TT^d \times \RR^d$ or $E=\TT^d$   is said to be of class ${\cal C}^{k,l}$ for $k\in \{0,1\}$ and $l\in \NN\cup \{ \infty\}$ if it has continuous derivatives up to order $k$ in $t\in [0,T]$ and up to order $l$ in $y\in E$ (or of all order if $l=\infty$).
For functions $ \phi: E \to \RR $ the notation ${\cal C}^{l}$ is used analogously.
\end{itemize}
In order to lighten the notations,  the dependency in $(t,x,u)$  or $(t,x)$ of  functions appearing inside equations will be omitted when no ambiguity is possible.

\subsection{The  Lagrangian stochastic model coupled with a Poisson equation}\label{sect:Poisson}

We start by establishing connections  between conditions related to the homogeneous mass distribution constraint, which are valid in arbitrary dimension:
\begin{lemma}\label{lem:relatedconditions}
Assume that $f$ is  a classical solution to equations \eqref{eq:pde_Fokker-Planck}  and \eqref{eq:pde_init_cond} for some function  $P:[0,T]\times\TT^{d}
 \rightarrow\er^{d}$  of class ${\cal C}^{0,2}$. Moreover, assume that
$$\rho(t,x):= \int_{\er^d} f(t,x,u) du, \quad V(t,x): = \int_{\er^d}u  f(t,x,u) du$$
are  functions  of class ${\cal C}^{1,1}$ in $[0, T]\times  \TT^d$,  that $\int_{\er^d }u^2 | D^{m} f(t,x,u)|du<+\infty $  for each  multiindex $|m| \leq 2$ and each $(t,x)\in [0,T]\times \TT^d$ (where $D$ is the derivative operator),
and further that  for all  $t\in [0,T]$ the function
$$x\mapsto \int_{\er^d} v_i v_j  f(t,x,v)dv$$
is of class ${\cal C}^2$.  Then,  the following system of equations is satisfied for  $(t,x)\in (0,T]\times \TT^d$:
\begin{equation*}
\left\{
\begin{aligned}
& \partial_{t}\rho + \nabla_{x}\cdot V=0,\\
&\partial_{t}(\nabla_{x}\cdot  V) +\beta \nabla_x\cdot V + \nabla_x\cdot \big( \rho \left(\nabla_x P -\beta\alpha V\right) \big) + \sum_{i,j=1}^{d} \partial_{x_i x_j} \int_{\er^d} v_i v_j  f(t,x,v)dv =0 \\
\end{aligned}
\right.
\end{equation*}
We deduce:
\begin{itemize}
\item[a)]$\rho(t,x)=\rho(0,x)$ for every $(t,x) \in [0,T]\times \TT^d$ if and only if $\nabla_x\cdot V(t,x)=0$  for every $(t,x) \in [0,T]\times  \TT^d$.
\item[b)]$\nabla_x\cdot V(t,x)=e^{-\beta t} \nabla_x\cdot V(0,x)$  for every $(t,x) \in [0,T]\times \TT^d$ if and only if  $P$ satisfies the  equation of elliptic type:
$$\nabla_x\cdot \big( \rho(t,x) \left(\nabla_x P(t,x) -\beta\alpha V(t,x)\right) \big) =-\sum_{i,j=1}^{d} \partial_{x_i x_j} \int_{\er^d} v_i v_j  f(t,x,v)dv \,  , \quad   (t,x)\in (0,T]\times  \TT^d.$$
\item[c)]If in  addition to \eqref{eq:pde_Fokker-Planck} and \eqref{eq:pde_init_cond},   condition \eqref{eq:pde_mass_constraint} is verified, then $P(t,x)$ is a solution to the Poisson equation
\begin{equation*}
\triangle_x P(t,x)= -  \sum_{i,j=1}^{d}\partial_{x_i x_j} \int_{\er^d} v_i v_j  f(t,x,v)dv, \quad(t,x)\in (0,T]\times \TT^d.
\end{equation*}
\item[d)]Set $\bar{\rho}(t,x):=\rho(t,x)-1$. If in addition to \eqref{eq:pde_Fokker-Planck}  and \eqref{eq:pde_init_cond}  we assume that the Poisson equation in part c) holds,  we have:
\begin{description}
\item when $\alpha=1$, $\partial_{t}(\nabla_{x}\cdot  V)   + \nabla_x\cdot \big(\bar{ \rho}  \left(\nabla_x P-\beta V(t,x) \right) \big) =0$;
\item when $\alpha=0$, $\partial_{t}(\nabla_{x}\cdot  V)   + \nabla_x\cdot \big(\bar{ \rho} (t,x)  \nabla_x P  \big)+ \beta\nabla_x \cdot V =0$.
\end{description}
\end{itemize}
\end{lemma}
\begin{proof} The first equation is obtained by integrating equation \eqref{eq:pde_Fokker-Planck} with respect to $u\in \er^d$, and using the  assumptions in order  to  integrate   by parts and get rid of integrals of divergence type terms. To get the second equation, we first take the derivative with respect to the variable $x_i$ in equation \eqref{eq:pde_Fokker-Planck}, then multiply it by $u_i$ and sum over $i=1,\dots, d$, before integrating and proceeding as before. Statements a),b),c) and d) are then easily deduced.
\end{proof}

\begin{remark}
\begin{itemize}
\item[a)]According to Lemma \ref{lem:relatedconditions} part c), finding a solution to equation \eqref{eq:CVFP} requires in particular to find a solution to the highly singular Vlasov-Fokker-Planck equation \eqref{eq:VFPalternative}.
%

\item[b)]  If conditions  \eqref{eq:pde_Fokker-Planck}  and \eqref{eq:pde_init_cond} hold,  and the Poisson equation in Lemma \ref{lem:relatedconditions} part c) is satisfied, the
equation obtained in Lemma \ref{lem:relatedconditions} part d) together with the continuity equation $$\partial_{t}\bar{\rho}(t,x) + \nabla_{x}\cdot V(t,x)=0$$
furnish a system of two equations    that  the pair $(\bar{\rho},V)$ must satisfy. Thus, a strategy to prove in that situation that
  \eqref{eq:pde_mass_constraint} also holds    is  to prove that  such a system has the unique solution  $\bar{\rho}(t,x)= \nabla_{x}\cdot V(t,x) \equiv 0$ when starting from $(0,0)$.  We will be able to do this in the functional setting that we will consider, deducing thus a solution to \eqref{eq:CVFP} from a solution to  \eqref{eq:VFPalternative}.
\end{itemize}
\end{remark}
\section{Local analytic well-posedness in the vanishing kinetic potential case $(\beta = 0)$}\label{1d_beta=0}
In this section, we construct an analytical solution to the nonlinear Vlasov-Fokker-Planck equation associated with the incompressible Lagrangian stochastic model up to some small time horizon $T$, in the case $\beta=0$. Using the  weighted analytical functional space  introduced in \cite{JabNou-11} and a fixed point argument developed therein, we shall give in Theorem \ref{thm:ExistenceNonlinear} below a local-in-time well-posedness result for  the nonlinear Vlasov-Fokker-Planck equation:

\renewcommand{\theequation}{{VFP}}
\begin{equation}\label{eq:VFP}
\left\{
\begin{aligned}
&\partial_{t}f+u\partial_{x}f- \partial_{x}P\partial_{u}f -\frac{\sigma^2}{2}\partial^{2}_{u}f=0\mbox{ on }(0,T]\times\er^2,\\
&f(0,x,u)=f_{0}(x,u)\mbox{ on }\er^2,
\end{aligned}
\right.
\end{equation}
\renewcommand{\theequation}{\thesection.\arabic{equation}}
where
\begin{equation*}
 P(t,x)= - \int_{\er}u^{2}f(t,x,u)du, ~(t,x)\in [0, T]\times \er .
\end{equation*}

Notice that  periodicity is not  yet imposed.
Then, we will show in Corollary \ref{coro:UniformMass} that if  the obtained local  solution $f(t,x,u)$ of \eqref{eq:VFP} satisfies at $t=0$  the condition
\begin{description}
\item[\Hu{t}]:
\begin{enumerate}
\item[] $f(t,x,u) \mbox{ is }1-\mbox{periodic in }x$ for all $u\in \er$,
\item[]$\displaystyle \int_{\er}f(t,x,u)du=1$ for all $x\in\TT$ (\textit{Uniform  mass repartition in }$\TT$),
\item[]$\partial_{x}\int_{\er}uf (t,x,u)du=0$ for all $x\in\TT$ (\textit{Mean incompressibility in }$\TT$),
\end{enumerate}
\end{description}
it then satisfies {\it a fortiori}  the same properties for all $t\in [0,T]$.  To establish the latter result, the choice of analytical functional spaces and the use of the analytic norms in \cite{JabNou-11} will also be fundamental.

\subsection{The nonlinear Vlasov-Fokker-Planck equation in analytic spaces}
 We start by defining  the functional spaces where an equivalent version of equation \eqref{eq:VFP} will be studied.

 A function  $\er^2\ni (x,u)\mapsto \psi(x,u) \in \er$ having bounded derivatives   of all order   is said to be analytic  if there exists $C>0$ and some $\bar{\lambda}>0$ such that for all $k,l\in \mathbb{N}$,

  $$\Vert\partial^{k}_{x}\partial^{l}_{u}\psi\Vert_{\infty}\leq C\frac{  k!l!}{\bar{\lambda}^{k+l}}, $$
  where $\Vert\cdot\Vert_{\infty}$ is the uniform norm on $\er^2$, and where the convention $0!=1$ is used. For such  functions and a general $\lambda>0$, we introduce the  analytic norm:
\begin{equation*}
\Vert \psi\Vert_{\lambda}:=\sum_{k,l\in\NN}\frac{\lambda^{k+l}}{k!l!}\left\|\partial^{k}_{x}\partial^{l}_{u}\psi\right\|_{\infty}
\end{equation*}
 and observe that  $\|\psi\|_\lambda$ is finite whenever $\lambda <\bar{\lambda}$.
We further introduce the $\lambda$-derivatives of these norms for each order $a\in\NN$,
\begin{align*}
\Vert \psi\Vert_{\lambda,a}:=\frac{d^{a}}{d\lambda^{a}}\Vert \psi\Vert_{\lambda} =
\sum_{k+l\geq a}\frac{(k+l)!}{(k+l-a)!}\frac{\lambda^{k+l-a}}{k!l!}\left\|\partial^{k}_{x}\partial^{l}_{u}\psi\right\|_{\infty}.
\end{align*}
Notice that $\Vert\psi\Vert_{\lambda,0}=\Vert \psi\Vert_{\lambda}$. We then define   respectively a norm and a seminorm by
\begin{equation*}
\Vert \psi\Vert_{\Hh,\lambda}:=\sum_{a\in\NN}\frac{1}{(a!)^{2}}\Vert \psi\Vert_{\lambda,a},
\quad\quad
\Vert \psi\Vert_{\tHh,\lambda}:=\sum_{a \geq 1}\frac{a^{2}}{(a!)^{2}}\Vert \psi\Vert_{\lambda,a}.
\end{equation*}
Last, we define the functional spaces associated with  $\Vert~\Vert_{\Hh,\lambda}$ and $\Vert ~\Vert_{\tHh,\lambda}$ :
\begin{subequations}
\begin{align}
\label{formula:SolutionSpacesPreliminary}
&\Hh(\lambda):=\left\{\psi \in\Cc^{\infty}(\er^2)\,\mbox{such that}\,\Vert \psi \Vert_{\Hh,\lambda}<+ \infty\right\},\\
&\widetilde\Hh(\lambda):=\left\{\psi \in\Cc^{\infty}(\er^2)\,\mbox{such that}\,\Vert \psi\Vert_{\tHh,\lambda}<+\infty\right\}.
\end{align}
\end{subequations}
The next two lemmas giving some insight about these (semi)norms  and will be useful later on:
\begin{lemma}\label{lem:NormCriterion}
Let $v:\er^2\to \er $ of class ${\cal C}^{\infty}$ be
  such that  $\Vert \partial^{l}_x   \partial^{k}_u v \Vert_{\infty}\leq \frac{C(k+m)! (l+n)!}{\overline{\lambda}^{k+l}}$   for some  $C,\bar{\lambda}>0$, some $m,n,j\in \NN$ and all $k,l \geq j$.
\begin{itemize}
\item[a)]If the previous holds for $j=0$,  then $v\in \Hh(\lambda)$  for all $\lambda\in [0,\overline{\lambda})$.
\item[b)]If the previous holds for $j=1$,  then $v\in \tHh(\lambda)$ for all $\lambda\in [0,\overline{\lambda})$.
\end{itemize}
\end{lemma}

\begin{proof} For $a\geq j$ and $\lambda\in[0,\overline{\lambda})$ we obtain from the assumption that
\begin{align*}
 \frac{d^{a}}{d\lambda^{a}}\Vert v\Vert_{\lambda,0}\leq &
 \frac{C}{\overline{\lambda}^{a}} \sum_{k+l\geq a}
\frac{(k+l)!(k+m)!(l+n)!}{k! l! (k+l -a )!}\left(\lambda/\overline{\lambda}\right)^{(k+l-a)}\\
&=\frac{C}{\overline{\lambda}^{a}} \sum_{k+l\geq a+m+n ,k\geq m,l\geq n}
\frac{(k+l-(m+n))!k!l !}{(k-m)! (l-n)! (k+l -(a+m+n) )!} \left(\lambda/\overline{\lambda}\right)^{(k+l-(a+m+n))}
\end{align*}
changing indexes $k+m$ to $k$ and $l+n$ to $l$.  Since
$\frac{(k+l-(m+n))!k!l !}{(k-m)! (l-n)! (k+l)!}\leq 1$  for $k\geq m, l\geq n$, we deduce that
\begin{align*}
\frac{d^{a}}{d\lambda^{a}}\Vert v\Vert_{\lambda,0}
\leq & \frac{C}{\overline{\lambda}^{a}} \sum_{k+l\geq a+m+n}
\frac{(k+l)!}{(k+l -(a+m+n) )!}\left(\lambda/\overline{\lambda}\right)^{(k+l-(a+m+n))}
=\frac{C}{\overline{\lambda}^{a}}  \left[\frac{d^{a+m+n}}{dr^{a+m+n}}\left(  \sum_{k,l\in \NN} r^{k+l} \right) \right]\bigg{\vert}_{r=\lambda/\overline{\lambda}}.
\end{align*}
Observing that for $r\in[0,1)$, $\sum_{k,l\in\NN}r^{k+l} =(\sum_{j\in \NN}r^{j})^2 =\frac{1}{(1-r)^2}$ and $|\frac{d^{a}}{dr^{a}}\frac{1}{(1-r)^2}| \leq (1+a)!$, we conclude that
\begin{align*}
\| v\|_{ \Hh, \lambda}= &
\sum_{a=0}^{\infty}\frac{1}{(a!)^{2}}\frac{d^{a}}{d\lambda^{a}}\Vert v \Vert_{\lambda,0}\leq C\sum_{a=0}^{\infty}\frac{1}{\overline{\lambda}^{a}}\frac{(a+1)\cdots(a+m+n+1)}{a!}<+\infty\mbox{  and}\\
\|v\|_{ \tHh, \lambda}= &
\sum_{a=1}^{\infty}\frac{a^2}{(a!)^{2}}\frac{d^{a}}{d\lambda^{a}}\Vert v \Vert_{\lambda,0}\leq \frac{C}{\overline{\lambda}} \sum_{a=0}^{\infty} \frac{1}{\overline{\lambda}^{a}}\frac{(a+1)\cdots (a+m+n+2 )}{a!} <+\infty.
\end{align*}
\end{proof}

\begin{remark}

Our main results below deal with initial data in the spaces  $\Hh(\lambda)$ and $\tilde{\Hh}(\lambda)$. By part a) of Lemma \ref{lem:NormCriterion}, for each analytic function $f$ such that $\|f\|_{\lambda}<+\infty$  for some  $\lambda>0$, one can find some $\lambda'>0$ such that $f\in \Hh(\lambda')$. Similarly, part b) of Lemma \ref{lem:NormCriterion} shows that if $\|f\|_{\lambda,1}<+\infty$  for some  $\lambda>0$,  then $f\in \tilde{\Hh}(\lambda'')$ for some other $\lambda''>0$.  Therefore, our results  will cover a large class of  analytic initial data.  The reason why spaces $\Hh(\lambda)$ and $\tilde{\Hh}(\lambda)$ are useful  here is that they will offer a more precise control on the convergence near the radius of analyticity as time varies.

\end{remark}

\begin{lemma}\label{lem:normprop} Let $\psi$ be an analytic function defined on $\er^2$. Then:
\begin{itemize}
\item[(i)]For each $a\in \NN$ one has $\Vert \psi\Vert_{\lambda,a+1}=\Vert\partial_{x}\psi\Vert_{\lambda,a}+\Vert\partial_{u}\psi\Vert_{\lambda,a}.$ We deduce that
\begin{equation*}
\Vert \psi\Vert_{\tHh,\lambda}=\Vert \partial_{x} \psi\Vert_{\Hh,\lambda}+\Vert \partial_{u} \psi\Vert_{\Hh,\lambda}.
\end{equation*}
\item[(ii)]Moreover,
\begin{equation*}
\frac{d}{d\lambda}\Vert \psi\Vert_{\Hh,\lambda}=\Vert \psi\Vert_{\tHh,\lambda}.
\end{equation*}
\item[(iii)]Last, for any pair  $\psi_{1},\psi_{2}$  of  analytic functions defined on $\er^2$\begin{equation*}\label{lem:normprop_bis}
\Vert\psi_{1} \psi_{2}\Vert_{\lambda}\leq \Vert\psi_{1}\Vert_{\lambda}\Vert\psi_{2}\Vert_{\lambda}.
\end{equation*}
\end{itemize}
\end{lemma}
\begin{proof} (i). The first identity follows from
\begin{align*}
\Vert\psi\Vert_{\lambda,a+1}=\frac{d^{a+1}}{d\lambda^{a+1}}\Vert \psi\Vert_{\lambda, 0}&=\sum_{m+l\geq a+1}\frac{(m+l)\cdots (m+l-a-1)\lambda^{m+l -a-1}}{m!l!}\Vert \partial^{m}_{x}\partial^{l}_{u}\psi\Vert_{\infty}\\
&=\sum_{m+l\geq a+1 ,l\geq 1}\frac{ (m+l-1)\cdots (m-a+l -1)\lambda^{m -a+l -1}}{m!(l-1)!}\Vert \partial^{m}_{x}\partial^{l}_{u}\psi\Vert_{\infty} \\
& \quad +
\sum_{m+l\geq a+1 ,m \geq 1 }\frac{ (m- 1 +l)\cdots (l-a+ m -1)\lambda^{l -a+m -1} }{(m-1)!l!}\Vert \partial^{m}_{x}\partial^{l}_{u}\psi\Vert_{\infty},
\end{align*}
by respectively changing  the indexes $l$ to $l+1$ and  $m$ to $m+1$ in the first and second sums  in the last expression. Multiplying   by $\frac{a^2}{(a!)^2}$  both sides of the previously established identity and summing the resulting expressions over $a\geq 1$ yields the identity for $\Vert\psi\Vert_{\tHh,\lambda}$. (ii) readily  follows from
\begin{align*}
\frac{d}{d\lambda}\Vert \psi\Vert_{\Hh,\lambda} =
\frac{d}{d\lambda}\sum_{a\in\NN}\frac{1}{(a!)^{2}}\|\psi\|_{\lambda,0} =\sum_{a\in\NN}\frac{1}{(a!)^{2}}\|\psi\|_{\lambda,a+1} = \sum_{a\geq 1}\frac{1}{((a-1)!)^{2}}\|\psi\|_{\lambda,a} = \sum_{a\geq 1}\frac{a^2}{(a!)^{2}}\|\psi\|_{\lambda,a}  = \Vert\psi\Vert_{\tHh,\lambda}.
\end{align*}
Finally, since $ \Vert\partial^{k}_{x}\partial^{l}_{u}\left(\psi_{1}\psi_{2}\right)\Vert_{\infty}\leq \sum_{r=0}^{k}\sum_{n=0}^{l}C^{r}_{k}C^{n}_{l}
\Vert\partial^{r}_{x}\partial^{n}_{u}\psi_{1}\Vert_{\infty}\Vert\partial^{k-r}_{x}\partial^{l-n}_{u}\psi_{2}\Vert_{\infty}$, we have
\begin{align*}
\Vert\psi_{1}\psi_{2}\Vert_{\lambda,0}
=\sum_{k,l\in\NN}\frac{\lambda^{k+l}}{k!l!}\Vert\partial^{k}_{x}\partial^{l}_{u} \left(\psi_{1}\psi_{2}\right)\Vert_{\infty}
&\leq
\sum_{r,n\in\NN}\Vert\partial^{r}_{x}\partial^{n}_{u}\psi_{1}\Vert_{\infty}\sum_{k\geq r}\sum_{l \geq n}\frac{C^{r}_{k}C^{n}_{l}\lambda^{k+l}}{k!l!}
\Vert\partial^{k-r}_{x}\partial^{l-n}_{u}\psi_{2}\Vert_{\infty}\\
&\leq \sum_{r,n\in\NN}\frac{\lambda^{r+n}}{r!n!}\Vert\partial^{r}_{x}\partial^{n}_{u}\psi_{1}\Vert_{\infty}\sum_{k\geq r}
\sum_{l\geq n}\frac{\lambda^{(k-r)+(l-n)}}{(k-r)!(l-n)!}
\Vert\partial^{k-r}_{x}\partial^{l-n}_{u}\psi_{2}\Vert_{\infty}
\end{align*}
which provides (iii) by  changing  the indexes $k$ to $k+r$ and $l$ to $l+n$ in the inner sums.
\end{proof}

We  now observe that finiteness of the analytical norm of a solution $f$ to \eqref{eq:VFP} is not enough to  provide a control of the function $(t,x)\mapsto\int_{\er}u^2 f(t,x,u)du$.
This is the reason why  we  introduce  a weight function intended to truncate  the velocity state space in a suitable sense.  More precisely, assume that  $f:[0,T]\times \er^2 \to \er$ is a ${\cal C}^{1,\infty}$  solution of equation \eqref{eq:VFP} with bounded derivatives of all order, and set
\begin{equation}\label{formula:TruncatedSolution}
g(t,x,u):=\omega(u)f(t,x,u),
\end{equation}
where $\omega:\er\to(0,+\infty)$ is a  weight function such that $\int_{\er}\frac{u^{2}}{\omega(u)}du<+\infty.$
Then, the regularity of velocity moments of $f$  is easily controlled in terms of the regularity  of $g$:
\begin{align*}
\sup_{(t,x)\in[0,T]\times\er }\left|\partial^{k}_{x} \int_{\er}u^2f(t,x,u)\,du\right|
=&\sup_{(t,x)\in[0,T]\times\er }\left|\int_{\er}\frac{u^2}{\omega(u)}\partial^{k}_{x}g(t,x,u)\,du\right|\\
&\leq\sup_{(t,x,u)\in[0,T] \times\er^2}\left|\partial^{k}_{x}g(t,x,u)\right|\int_{\er} \frac{u^{2}}{\omega(u)}du.
\end{align*}
Moreover, since $\partial_{u}f=\partial_{u}g-g(\partial_{u}\ln\omega)$ and
$\omega\partial^{2}_{u}f=\partial^{2}_{u}g -2(\partial_{u}\ln(\omega))(\partial_{u}g)+g(\frac{2|\partial_{u}\omega|^{2}}{\omega^{2}}
-\frac{\partial^{2}_{u}\omega}{\omega})$, the function $g$  defined in \eqref{formula:TruncatedSolution} is  seen to satisfy the equation
\renewcommand{\theequation}{{VFP$\omega$}}
\begin{equation}\label{eq:VFPw}
\left\{
\begin{aligned}
&\partial_{t}g+ u \partial_{x}g - \left[
\partial_{x}P-\partial_{u}\ln\omega\right]\partial_{u}g -\frac{\sigma^2}{2}\partial^{2}_{u}g=g\partial_{x}P\partial_{u}\ln \omega -g h\mbox{ on }(0,T]\times\er^2,\\
&P(t,x)=-\int_{\er}\frac{u^{2}}{\omega(u)}g(t,x,u)du,  \\
&g(0,x,u)=g_{0}(x,u)\mbox{ on }\er^2,
\end{aligned}\right.
\end{equation}
\renewcommand{\theequation}{\thesection.\arabic{equation}}

%
 where
$$h(u):=\frac{\partial^{2}_{u}\omega(u)}{2\omega(u)}-|\partial_{u}\ln(\omega(u))|^{2} \, ; $$
reciprocally, given a solution  $g$ to \eqref{eq:VFPw}, the function $f$ defined by \eqref{formula:TruncatedSolution} is a solution to  \eqref{eq:VFP}.

In all the sequel, we shall assume that  $\omega:\er \to  (0,+\infty)$ is a  function of class ${\cal C}^{\infty}$ such that
\begin{description}
\item[\Hw]
\begin{enumerate}
\item[]$\displaystyle\lim_{|u|\to+\infty}  \frac{ \omega (u)}{|u|} =  +\infty$ and $\displaystyle\int_{\er}\frac{u^{2}}{\omega(u)}du=1$,
\item[]$\displaystyle{\limsup_{|u|\to\infty} \big| \frac{\omega'(u)}{\omega(u)} \big|<\infty, \,  \limsup_{|u|\to\infty} \big|\frac{\omega''(u)}{\omega'(u)}\big|<\infty}$,
\item[]Moreover, for some $\lambda_{0}>0$ we have  $ \ln(\omega)\in \tHh(\lambda_{0})$ and $h \in \Hh(\lambda_{0})$ .
\end{enumerate}
\end{description}
   The following result provides examples of such functions $\omega$, as well as tractable conditions on the initial condition  $f_0$  ensuring the type of  bounds on   $g_0$  required by our results on  equation \eqref{eq:VFPw}.  Its proof relies on  Lemma \ref{lem:NormCriterion} and is given in  Appendix \ref{sec:WeightAnalycity}.

\begin{lemma}\label{analyticw}
\begin{itemize}
\item[i)]  Let  $s\geq 4$ be  a positive integer.  Then, condition  \Hw  \, holds for the  weight function $\omega(u): = c(s) (1+u^2)^{\frac{s}{2}} $ for all value  $\lambda_0\in (0,\frac{1}{4})$,  where $c(s)>0$ is  such that $  \int_{\er}\frac{u^{2}}{\omega(u)}\,du=1$ .
\item[ii)] Let $f_0:\RR^2\to \RR$ be a function of class ${\cal C}^{\infty}$ such that  for some  even integer $s\geq 4 $, constants $C_0,\bar{\lambda}>0$, some $m,n,j\in \NN$ and all $k,l \geq j$, one has
\begin{equation}\label{condf_0d1}
 \Vert  (1+u^2)^{\frac{s}{2}} \partial^{l}_x   \partial^{k}_u f_0  \Vert_{\infty}\leq \frac{C_0(k+m)! (l+n)!}{\overline{\lambda}^{k+l}} .
 \end{equation}

 Then,  the function $g_0(x,u):=\omega (u) f_0(x,u)$ with $\omega (u) $ as in i) satisfies the assumptions of Lemma  \ref{lem:NormCriterion}  with $C:=C_0'= C_0 \kappa(s)e^{\bar{\lambda}}$ and    $\kappa(s)>0$  a bound for the absolute values of the coefficients of the polynomials $\omega,\partial_u\omega,\dots,\partial_u^s \omega$.  In particular,  if  for  some $n,m \in \NN$  condition \eqref{condf_0d1} holds for all $k,l\geq 0$, then for all  $\lambda\in [0,\overline{\lambda})$ one has
 \begin{equation*}
 \begin{split}
\| g\|_{ \Hh, \lambda} \leq &  C_0 \kappa(s)e^{\bar{\lambda}}   \mu(\overline{\lambda},m+n+1)\mbox{  and }\\
\|g\|_{ \tHh, \lambda}  \leq  & C_0 \kappa(s)\frac{e^{\bar{\lambda}}}{ \overline{\lambda} }  \mu(\overline{\lambda},m+n+2)
\end{split}
\end{equation*}
where
 $ \mu(\overline{\lambda}, p):= \sum_{a=0}^{\infty}\frac{1}{\overline{\lambda}^{a}}\frac{(a+1)\cdots(a+p)}{a!}<+\infty$ for all $p\in \NN, p\geq 1$.
 \end{itemize}
\end{lemma}

\subsection{Main results}

Given $K, T$ and $\lambda_{0}$ strictly positive real numbers such that $\lambda_0>T(1+K)$,
 and the function $$\lambda(t):=\lambda_{0}-(1+K)t,$$
we now define the spaces
\begin{align*}
&\Hh_{\lambda_0,K,T}:=\left\{\psi\in\Cc^{1,\infty}([0,T]\times\er^2)\mbox{ such that } \sup_{t\in[0,T]}\Vert \psi(t)\Vert_{\Hh,\lambda(t)}<+\infty \right\},\\
&\tHh_{\lambda_0,K,T}:=\left\{\psi\in\Cc^{1,\infty}([0,T]\times\er^2)\mbox{ such that }\int_{0}^{T}\Vert \psi(t)\Vert_{\tHh,\lambda(t)}dt< +\infty \right\}
\end{align*}
and their subsets defined for a positive constant $M$:
\begin{align*}
&{\cal B}^{M}_{\lambda_0,K,T}:=\left\{\psi\in \Hh_{\lambda_0,K,T} \mbox{ such that } \sup_{t\in[0,T]}\Vert \psi(t)\Vert_{\Hh,\lambda(t)}  \leq M\right\}, \\
&\widetilde{{\cal B}}^{M}_{\lambda_0,K,T}:=\left\{\psi\in \tHh_{\lambda_0,K,T}\mbox{ such that } \int_{0}^{T}\Vert \psi(t)\Vert_{\tHh,\lambda(t)}dt\leq M\right\}.
\end{align*}
We are ready to state the main result  of this section:
\begin{thm}\label{thm:ExistenceNonlinear}
Let $M, T$ be positive constants and $\omega:\er\to (0,+\infty)$ be a function of class ${\cal C}^{\infty}$  satisfying \Hw\, for some $\lambda_0>0$.  Introduce the finite constants  $\gamma_{0}:= \Vert\ln(\omega)\Vert_{\tHh,\lambda_{0}} $ and $\gamma_{1}:= \Vert h\Vert_{\Hh,\lambda_{0}} $ and  assume that
\begin{itemize}
\item[a)]$T<\frac{\lambda_0}{ 2 + \lambda_0 +4\gamma_0 }$,
\item[b)]$M \leq \frac{1}{16}(K - \lambda_{0}- 4\gamma_{0}- 1)$ for some $K$ in the nonempty interval $(1 + \lambda_0 +4\gamma_0 , \frac{\lambda_0}{T}  -1)$   and
\item[c)]$M(1+\gamma_{0})\exp\{(M\gamma_{0}+\gamma_{1})T\}<1$.
\end{itemize}
Assume moreover that $f_0:\er^2 \to \er$ is a function of class ${\cal C}^{\infty}$ and that $g_0(x,u):=\omega(u)f_0(x,u)$ satisfies
\begin{itemize}
\item[d)]$\max\{ \Vert g_{0}\Vert_{\Hh,\lambda_{0}} , T\Vert g_{0}\Vert_{\tHh,\lambda_{0}}\} \leq M$  and
\item[e)]$\Vert g_{0}\Vert_{\Hh,\lambda_{0}}\exp(T(\gamma_1+16\gamma_0)) \leq M \exp(-(16+\gamma_0)M)$.
\end{itemize}
Then,  equation \eqref{eq:VFPw} has a unique smooth solution  $g \in {\cal B}^{M}_{\lambda_0,K,T}  \cap \widetilde{{\cal B}}^{M}_{\lambda_0,K,T} $. In particular, under the previous assumptions, a solution $f\in {\cal C}^{1,\infty}$ to \eqref{eq:VFP} with initial condition $f_0$ exists.
\end{thm}

\begin{corollary}\label{coro:UniformMass}
Let  $f$ be the  solution to \eqref{eq:VFP} given in Theorem \ref{thm:ExistenceNonlinear} and assume that \Hu{0} holds. Then, $f(t,x,u)$ satisfies \Hu{t}, for all $t$ in $[0,T]$.
In particular, if   the assumptions of Theorem \ref{thm:ExistenceNonlinear} and condition \Hu{0} hold, then a solution to \eqref{eq:CVFP} with $\beta=0$ exists.
\end{corollary}

 \begin{remark}\label{rem:Tsmall}
For instance, let   $f_0:(\RR^d)^2\to \RR$   be a function of class $   {\cal C}^{\infty}$ and  $C_0,\bar{\lambda}>0$ , $n,m \in \NN$ be numbers satisfying  condition \eqref{condf_0d1}  for  every $k,l\geq 0$. Suppose moreover that for some $\lambda_0<\min\{\bar{\lambda},\frac{1}{4}\}$ one has
$$C_0< \kappa_0(\bar{\lambda},s):=\frac{1}{2\kappa(s)e^{\bar{\lambda}} \mu (\bar{\lambda}, m+n+1)}\frac{\ln 2}{(16+  \Vert\ln(\omega)\Vert_{\tHh,\lambda_{0}})  }$$
for $\omega$ as in Lemma   \ref{analyticw} i).
Setting $M:= 2C_0  \kappa(s)e^{\bar{\lambda}} \mu (\bar{\lambda}, m+n+1)  $ and  $\gamma_{0}= \Vert\ln(\omega)\Vert_{\tHh,\lambda_{0}} $, we then have
$$\kappa_1(C_0,\bar{\lambda},s):=  \min \left\{
\frac{ \lambda_0}{16 M  + \lambda_0 +4\gamma_0  +2} ,  \frac{2 \bar{\lambda} \mu(\overline{\lambda},m+n+1)}{ \mu(\overline{\lambda},m+n+2)} ,
-\frac{\ln ( M  (1+\gamma_0))   }{ M   \gamma_0 +\gamma_1}, \frac{\ln  2 - M (16+\gamma_0)}{\gamma_1+16\gamma_0}   \right\}>0. $$
Taking $T<\kappa_1(C_0,\bar{\lambda},s)$,  conditions  a)  and c) of Theorem \ref{thm:ExistenceNonlinear} are trivially satisfied, condition b) is satisfied (with equality) for   $K:= 16 M  + \lambda_0 +4\gamma_0  +1$, and conditions d) and e) hold because of  the estimates in Lemma  \ref{analyticw} ii).\end{remark}

The steps of the proof   of  Theorem \ref{thm:ExistenceNonlinear}  are the following: first we
will establish in Section \ref{sec:LinearApproximation} the existence  of an analytic solution to a suitable linear version of \eqref{eq:VFP} in a small time interval, along with useful estimates.
Then, under additional constraints we construct  in  Section \ref{sec:ConstructionSolutionNonlinear} a solution to the nonlinear equation \eqref{eq:VFP} by means of a fixed point argument.

Before proceeding, let us prove Corollary \ref{coro:UniformMass}: 
\begin{proof}[Proof of Corollary \ref{coro:UniformMass}] Periodicity of the solution is an easy consequence of the fixed point method
employed in the proof of Theorem \ref{thm:ExistenceNonlinear} (see remark \ref{rem:periodicity} in Section  \ref{sec:ConstructionSolutionNonlinear}).

Now, thanks to the assumptions on $\omega$ and the fact that $f(t,x,u)\omega(u)$ belongs to $\Hh(\lambda(t))$ for each $t\in [0,T]$, the assumptions of Lemma \ref{lem:relatedconditions} are satisfied (in particular the  integrals $\int_{\er}u\partial^2_u f(t,x,u)du=\int_{\er}\partial_u f(t,x,u)du=\int_{\er}\partial^2_u f(t,x,u)du$ exist and vanish;  moreover,  we have $\int_{\er}u\partial_u f(t,x,u)du =- \int_{\er} f(t,x,u)du$). Therefore, thanks to Lemma \ref{lem:relatedconditions} b),  the functions $\overline{\rho}(t,x):=\rho(t,x)-1=\int_{\er} f(t,x,u)du -1 $,  $V(t,x):=\int_{\er}u f(t,x,u)\,du$ and  $  P(t,x)= -   \int_{\er}u^{2}f(t,x,u)\,du$ satisfy the following system of equations: for all
$(t,x) \in (0,T]\times \er$,
\begin{equation*}
\left\{
\begin{aligned}
&\partial_{t}\overline{\rho}(t,x)=- \partial_{x}V(t,x),\\
&\partial_{t}(\partial_{x} V(t,x))=- \partial_{x}(\overline{\rho}(t,x)\partial_{x}P(t,x)).
\end{aligned}
\right.
\end{equation*}
 From the latter and from  Lemma \ref{lem:normprop}-(iii), we obtain, for each $\lambda\in [0,\lambda_0)$,
\begin{equation}
\label{proofst1:LocalProperties}
\left\{
\begin{aligned}
\partial_{t}\Vert\overline{\rho}(t)\Vert_{\lambda}&\leq \Vert  \partial_{x} V(t)\Vert_{\lambda},\\
\partial_{t}\Vert \partial_{x} V(t)\Vert_{\lambda}&\leq \Vert\partial_{x}P(t)\Vert_{\lambda}\Vert\partial_{x}\overline{\rho}(t)\Vert_{\lambda}+
\Vert\partial^{2}_{x}P(t)\Vert_{\lambda}\Vert \overline{\rho}(t)\Vert_{\lambda}.
\end{aligned}
\right.
\end{equation}
Since $\Vert\partial_{x}\overline{\rho}(t)\Vert_{\lambda}=\frac{d}{d\lambda}\Vert\overline{\rho}(t)\Vert_{\lambda}$ by  Lemma \ref{lem:normprop}-(ii), \eqref{proofst1:LocalProperties} rewrites as
\begin{equation*}
\left\{
\begin{aligned}
\partial_{t}A(t,\lambda)&\leq B(t,\lambda),\\
\partial_{t}B(t,\lambda)&\leq \Vert\partial_{x}P(t)\Vert_{\lambda}\partial_{\lambda}A(t,\lambda)+
\Vert\partial^{2}_{x}P(t))\Vert_{\lambda}A(t,\lambda),
\end{aligned}
\right.
\end{equation*}
for $A(t,\lambda):=\Vert\overline{\rho}(t)\Vert_{\lambda}$ and $B(t,\lambda):=\Vert \partial_{x} V(t)\Vert_{\lambda}$. Since  $t\mapsto \lambda(t)$ is decreasing and,  by Theorem \ref{thm:ExistenceNonlinear}, $P\in {\cal B}^{M}_{\lambda_0,K,T} \cap
\widetilde{{\cal B}}^{M}_{\lambda_0,K,T} $, we have
\begin{align*}
& \Vert\partial_{x}P(t)\Vert_{\lambda(T)}\leq
 \Vert\partial_{x}P(t)\Vert_{\lambda(t)}\leq \max_{s \in[0,T]}\Vert P(s)\Vert_{\Hh,\lambda(s)}\leq M,\\
&  \Vert\partial^{2}_{x}P(t)\Vert_{\lambda(T)}\leq
\Vert\partial^{2}_{x}P(t)\Vert_{\lambda(t)}\leq 4\max_{s \in[0,T]}\Vert P(s)\Vert_{\Hh,\lambda(s) }\leq 4M.
\end{align*}
We deduce that for all $\lambda\in[0,\lambda(T)]$, $t\in [0,T]$,
\begin{equation}
\label{proofst3:LocalProperties}
\left\{
\begin{aligned}
\partial_{t}A(t,\lambda)&\leq B(t,\lambda),\\
\partial_{t}B(t,\lambda)&\leq M\partial_{\lambda}A(t,\lambda)+4M A(t,\lambda)
\end{aligned}
\right.
\end{equation}
 because  $\partial_{\lambda}A(t,\lambda)\geq 0$.  Now set $\mathcal{Y}(t,\lambda):= A(t,\lambda)+bB(t,\lambda)$ where $b$ is a positive constant that we will specified later. Since also   $\partial_{\lambda}B(t,\lambda)\geq 0$,  from \eqref{proofst3:LocalProperties} we obtain
\begin{equation*}
\partial_{t}\mathcal{Y}(t,\lambda)\leq   B(t,\lambda)+ 4 b M  A(t,\lambda)+ b M \partial_{\lambda}A(t,\lambda)\leq \left(\frac{1}{b} \vee 4 b M \right) \mathcal{Y}(t,\lambda)+b M\partial_{\lambda}\mathcal{Y}(t,\lambda).
\end{equation*}
That is, with $b_2:= b M>0$ and  $b_1:=\left(\frac{1}{b} \vee 4 b M\right)>0$, it holds that
\begin{equation*}
\partial_{t}\mathcal{Y}(t,\lambda)\leq  b_{1} \mathcal{Y}(t,\lambda)+b_{2}\partial_{\lambda}\mathcal{Y}(t,\lambda),~\forall t\in[0,T],~\forall\lambda\in[0,\lambda(T)).
\end{equation*}
We now observe that  the function $t\mapsto \mathcal{Y}(t,\gamma(t))$ with  $\gamma(t):=\lambda(T)-b_{2}t$ is constant for all $t\in [0,\frac{\lambda(T)}{ b_{2}} )$. Indeed, we have
\begin{equation*}
\partial_{t}(\mathcal{Y}(t,\gamma(t)))=(\partial_{t}\mathcal{Y})(t,\gamma(t))-b_{2}\partial_{\lambda}\mathcal{Y}(t,\gamma(t))\leq b_{1}\mathcal{Y}(t,\gamma(t)),
\end{equation*}
 and Gronwall's lemma,  together with assumption \Hu{0} implying  that $\mathcal{Y}(0,\lambda)=0$ for all non negative $\lambda$,  yield $\mathcal{Y}(t,\gamma(t))=0$ for all $t\in [0,\frac{\lambda(T)}{b_{2}} )$. This shows that
$\overline{\rho}(t,x) = | \partial_{x}V(t,x)|=0$ for all $t\in [0,\frac{\lambda(T)}{b_{2}})$. Choosing $b=\lambda(T)/( M T)$, we  conclude the result,  using also the uniform bounds available   up to  time $t=T$.
\end{proof}

\subsection{The linearized equation}\label{sec:LinearApproximation}
Consider the linear equation
\renewcommand{\theequation}{{FP$\omega$}}

\begin{equation}\label{eq:FPw}
\left\{
\begin{aligned}
&\partial_{t}g+u \partial_{x}g -\left(
\partial_{x}Q-\partial_{u}(\ln \omega)\right)\partial_{u}g-\frac{\sigma^2}{2}\partial^{2}_{u}g=g\partial_{x}Q\partial_{u}\ln \omega +g h\mbox{ on }(0,T)\times\er^2,\\
&g(0,x,u)=g_{0}(x,u):=\omega(u) f_{0}(x,u)\mbox{ on }\er^2,
\end{aligned}\right.
\end{equation}
\renewcommand{\theequation}{\thesection.\arabic{equation}}
where $Q:[0,T]\times\er \to \er  $ is a given function, with uniformly in $t\in [0,T]$ bounded derivatives of all order in $x\in \er$. Equation \eqref{eq:FPw} is easily seen to be equivalent,   through the relation \eqref{formula:TruncatedSolution},  to the linear version of \eqref{eq:VFP}:
\renewcommand{\theequation}{{FP}}
\begin{equation}\label{eq:FP}
\left\{
\begin{aligned}
&\partial_{t}f+u \partial_{x}f- \partial_{x}Q\partial_{u}f -\frac{\sigma^2}{2}\partial^{2}_{u}f=0\mbox{ on }(0,T)\times\er^2\\
&f(0,x,u)=f_{0}(x,u)\mbox{ on }\er^2.
\end{aligned}\right.
\end{equation}
\renewcommand{\theequation}{\thesection.\arabic{equation}}
Existence and uniqueness of a ${\cal C}^{\infty}$-solution to the two previous equations is  recalled in Theorem \ref{thm:WellposedClassicLinearVFP} in Appendix \ref{sec:WellPosedLinearVFP}.  We next prove that the solution $g$ to \eqref{eq:FPw} is indeed analytic  whenever the inputs $g_0$ and $Q$ have
small enough analytic norms and  the time horizon $T>0$ is small enough:

\begin{thm}\label{thm:Closure}
Assume that  for some $\lambda_0>0$ condition  \Hw  \,  holds,  and that $g_0:\er^2\to \er$ is a function of class ${\cal C}^{\infty}$ such that $\Vert g_{0}\Vert_{\Hh,\lambda_{0}}<+\infty.$
For $\gamma_0$ and $\gamma_1$ as in Theorem \ref{thm:ExistenceNonlinear}, let $T>0$ and $M_1$>0 be a   time  horizon  and a   constant satisfying
\begin{itemize}
\item[a)]$T<\frac{\lambda_0}{ 2 + \lambda_0 +4\gamma_0 }$ and
\item[b)]$M_{1} \leq \frac{1}{16}(K - \lambda_{0}- 4\gamma_{0}- 1)$  for some $K$ in the nonempty
set $(1 + \lambda_0 +4\gamma_0 , \frac{\lambda_0}{T}  -1)$.
\end{itemize}
Then, for any   $M_2>0$ and $Q\in  {\cal B}^{M_1}_{\lambda_0,K,T} \cap
\widetilde{{\cal B}}^{M_2}_{\lambda_0,K,T} $, equation \eqref{eq:FPw} has  a solution $g$ of class $\Cc^{1,\infty}$ such that $$g \in  {\cal B}^{\hat{M}}_{\lambda_0,K,T} \cap
\widetilde{{\cal B}}^{\hat{M}}_{\lambda_0,K,T}  $$
where $\hat{M}=\Vert g_{0}\Vert_{\Hh,\lambda_{0}}  \exp\left\{ T (\gamma_{1}+16 \gamma_{0})+(16+\gamma_{0})M_2  \right\}$.
\end{thm}

In the proof we need to deal with truncated versions of the analytic norms previously introduced.
For an arbitrary function  $\psi$ of class $\Cc^{\infty}$ and a fixed $A\in\NN$, set
\begin{equation*}
\setA:=\{0,\cdots,A\}, \qquad \Vert \psi\Vert_{\lambda;A}:=\sum_{k,l\in\setA}\frac{\lambda^{k+l}}{k!l!}\left\|\partial^{k}_{x}\partial^{l}_{u}\psi\right\|_{\infty},
\end{equation*}
\begin{equation*}
\Vert \psi\Vert_{\lambda,a;A}:=\frac{d^{a}}{d\lambda^{a}}\Vert \psi\Vert_{\lambda;A} =
\sum_{k,l\in \setA; k+l\geq a}\frac{(k+l)!}{(k+l-a)!}\frac{\lambda^{k+l-a}}{k!l!}\left\|\partial^{k}_{x}\partial^{l}_{u}\psi\right\|_{\infty},
\end{equation*}
\begin{equation*}
\Vert \psi\Vert_{\Hh,\lambda;A }:=\sum_{a\in\setA}\frac{1}{(a!)^{2}}\Vert \psi\Vert_{\lambda,a;A}, \qquad
\Vert \psi\Vert_{\tHh,\lambda;A}:=\sum_{a\in\setA}\frac{a^{2}}{(a!)^{2}}\Vert \psi\Vert_{\lambda,a;A}.
\end{equation*}

Using a maximum principle for kinetic Fokker-Planck equation, stated in Appendix \ref{sec:WellPosedLinearVFP},  we start the proof by  establishing estimates for the time evolution of the norms $\Vert g(t)\Vert_{\Hh,\lambda(t);A}$ and $\Vert g(t)\Vert_{\tHh,\lambda(t); A}$  along a solution  $g$ of the linear equation \eqref{eq:FPw},
in terms of $\Vert Q(t)\Vert_{\Hh,\lambda(t)}$, $\Vert \partial_{u}\ln(\omega)\Vert_{\Hh,\lambda(t)}$,
 $\Vert h\Vert_{\tHh,\lambda(t)}$ and $\Vert Q(t)\Vert_{\tHh,\lambda(t)}$.

\subsubsection{Regularity estimates}
Let $g$ be a smooth solution to \eqref{eq:FPw}. Observe that, for all $(t,x,u)\in[0,T]\times\er^{2}$, we have the identities
\begin{align*}
\partial^{k}_{x}\partial^{l}_{u}(u\partial_{x}g(t,x,u))&=u\partial^{k+1}_{x}\partial^{l}_{u}g(t,x,u) +\1_{\{l\geq 1\}} l\partial^{k+1}_{x}\partial^{l-1}_{u}g(t,x,u),\\
\partial^{k}_{x}\partial^{l}_{u}(\partial_{x}Q(t,x)\partial_{u}g(t,x,u))&=
\sum_{m=0}^{k}C^{m}_{k}(\partial^{m+1}_{x}Q(t,x))(\partial^{k-m}_{x}\partial^{l+1}_{u}g(t,x,u))\\
&= \partial_{x}Q(t,x)\partial^{k}_{x}\partial^{l+1}_{u}g(t,x,u)+\1_{\{k\geq 1\}}
\sum_{m=0}^{k-1}C^{m}_{k}\partial^{k-m+1}_{x}Q(t,x)\partial^{m}_{x}\partial^{l+1}_{u}g(t,x,u),\\
\partial^{k}_{x}\partial^{l}_{u}\left(\partial_{u}\ln(\omega(u)) \partial_{u}g(t,x,u)\right)&=
\sum_{n=0}^{l} C^{n}_{l} \partial^{l-n+1}_{u}\ln(\omega(u)) \partial^{n+1}_{u}\partial^{k}_{x}g(t,x,u)\\
&=\partial_{u}\ln \omega(u)\partial^{l+1}_{u}\partial^{k}_{x}g(t,x,u) +
\1_{\{l \geq 1\}}\sum_{n=0}^{l-1}C^{n}_{l}\,\partial^{l-n +1}_{u}\ln \omega(u)\,
\partial^{n+1}_{u}\partial^{k}_{x}g(t,x,u),\\
\partial^{k}_{x}\partial^{l}_{u}\left(\partial_{x}Q(t,x)\partial_{u}\ln(\omega(u)) g(t,x,u)\right)&=
\sum_{n=0}^{l}\sum_{m=0}^{k}C^{n}_{l}C^{m}_{k}\partial^{k-m +1}_{x}Q(t,x)\,\partial^{l-n+1}_{u}\ln \omega(u)
\,\partial^{m}_{x}\partial^{n}_{u}g(t,x,u), \mbox{ and }\\
\partial^{k}_{x}\partial^{l}_{u}(g(t,x,u)h(u))&=\sum_{n=0}^{l}C^{n}_{l}\partial^{k}_{x}\partial^{n}_{u}g(t,x,u)
\,\partial^{l-n}_{u}h(u).
\end{align*}
By applying the differential operator $\partial^{k}_{x}\partial^{l}_{u}$ to \eqref{eq:FPw}, we deduce,that
\begin{align*}
&\partial_{t}(\partial^{k}_{x}\partial^{l}_{u}g)+u\partial_{x} (\partial^{k}_{x}\partial^{l}_{u} g)-
\left(\partial_{x}Q-\partial_{u}\ln \omega\right) \partial_{u} (\partial^{k}_{x}\partial^{l}_{u} g)
-\frac{\sigma^2}{2}\partial^{2}_{u}(\partial^{k}_{x}\partial^{l}_{u}g)\\
&=-l\partial^{k+1}_{x}\partial^{l-1}_{u}g\1_{\{l\geq 1\}}  +  \1_{\{k\geq 1\}} \sum_{m=0}^{k-1}C^{m}_{k}\partial^{k-m}_{x}\partial_{x}Q\,
\partial^{m}_{x}\partial^{l +1}_{u}g
- \1_{\{l\geq 1\}} \sum_{n=0}^{l-1}C^{n}_{l} \partial^{l-n+1}_{u}\ln\omega\partial^{n+1}_{u}\partial^{k}_{x}g\\
&\quad+\sum_{n=0}^{l}C^{n}_{l} \partial^{k}_{x}\partial^{n}_{u}g\,\partial^{l-n}_{u}h
+\sum_{n=0}^{l}\sum_{m=0}^{k}C^{n}_{l}C^{m}_{k} \partial^{k-m+1}_{x}Q
\partial^{l-n+1}_{u}\ln\omega \partial^{m}_{x}\partial^{n}_{u}g.
\end{align*}

The function $\partial_{x}^{k}\partial_{u}^{l}g$ is thus a classical solution to a linear Fokker-Planck equation. Applying the maximum principle stated in Theorem \ref{thm:WellposedClassicLinearVFP} in the appendix section \ref{sec:WellPosedLinearVFP}, we deduce  that
\begin{equation}
\label{estim:AnalyticalEstim1d}
\begin{aligned}
\frac{d}{dt}\Vert\partial^{k}_{x}\partial^{l}_{u}g(t)\Vert_{\infty}&\leq  \1_{\{l\geq 1\}}   l\Vert\partial^{k+1}_{x}\partial^{l-1}_{u}g(t)\Vert_{\infty} +\1_{\{k\geq 1\}}
\sum_{m=0}^{k-1}C^{m}_{k}\Vert \partial^{m}_{x}\partial^{l+1}_{u}g(t)\Vert_{\infty}\Vert \partial^{k-m+1}_{x}Q(t)\Vert_{\infty}\\
&\quad +\1_{\{l\geq 1\}}\sum_{n=0}^{l-1}C^{n}_{l}\Vert\partial^{n+1}_{u}\partial^{k}_{x}g(t)\Vert_{\infty}\Vert\partial^{l-n+1}_{u}\ln(\omega)\Vert_{\infty}
+\sum_{n=0}^{l}C^{n}_{l}\Vert\partial^{k}_{x}\partial^{n}_{u}g(t)\Vert_{\infty}\Vert\partial^{l-n}_{u}h\Vert_{\infty}\\
&\quad+\sum_{n=0}^{l}\sum_{m=0}^{k}C^{n}_{l}C^{m}_{k}\Vert\partial^{m}_{x}\partial^{n}_{u}g(t)\Vert_{\infty}\Vert\partial^{k-m+1}_{x}Q(t)\Vert_{\infty}
\Vert\partial^{l-n+1}_{u}\ln(\omega)\Vert_{\infty}.
\end{aligned}
\end{equation}
We now obtain estimates for the function $t\mapsto\Vert g(t)\Vert_{\lambda,a ;A}$ for fixed $\lambda>0$ and $A\in\NN$.
\begin{lemma}\label{lem:AnalyticalEstim1d} For each $A\in\NN$,  $a\in\setA=\{0,...,A\}$ and $\lambda>0$, a smooth solution $g$ to \eqref{eq:FPw} satisfies:
\begin{equation*}
\begin{aligned}
\frac{d}{dt}\Vert g(t) \Vert_{\lambda,a;A}&\leq\lambda\Vert g(t)\Vert_{\lambda,a+1; A}+a \Vert g(t)\Vert_{\lambda,a;A}+
\frac{d^{a}}{d\lambda^{a}}\bigg(\Vert g(t)\Vert_{\lambda,1;A}\big\{\Vert Q(t)\Vert_{\lambda,1;A}+\Vert \ln(\omega)\Vert_{\lambda,1;A}
\big\}\bigg)\\
&\quad+\frac{d^{a}}{d\lambda^{a}}\bigg(\Vert g(t)\Vert_{\lambda,0;A}\big\{\Vert h\Vert_{\lambda,0}+\Vert Q(t)\Vert_{\lambda,1;A}
\Vert \ln(\omega)\Vert_{\lambda,1;A}\big\}\bigg).
\end{aligned}
\end{equation*}
\end{lemma}
\begin{proof}
Multiplying both sides of the inequality \eqref{estim:AnalyticalEstim1d} by
$\frac{d^{a}}{d\lambda^{a}}\frac{\lambda^{k+l}}{k!l!} =\frac{(k+l)!\lambda^{k+l-a}}{(k+l-a)!k!l!}\1_{\{k+l\geq a\}}$ and  summing over $k,l\in\setA$ with $k+l\geq a$, we get
\begin{equation}
\label{proofst4:AnalyticalEstim}
\begin{aligned}
&\frac{d}{dt}\Vert g(t) \Vert_{\lambda,a;A}=
\sum_{k,l\in\setA: k+l\geq a}\frac{d^{a}}{d\lambda^{a}}\frac{\lambda^{k+l}}{k!l!}\frac{d}{dt}\Vert\partial^{k}_{x}\partial^{l}_{u}g(t)\Vert_{\infty}\\
&\leq  \sum_{k,l\in\setA: k+l\geq a, l\geq 1}\frac{d^{a}}{d\lambda^{a}}\frac{l\lambda^{k+l}}{k!l!}\Vert\partial^{k+1}_{x}\partial^{l-1}_{u}g(t)\Vert_{\infty}+
\sum_{k,l\in\setA:k+l\geq a,k\geq 1}\frac{d^{a}}{d\lambda^{a}}\frac{\lambda^{k+l}}{k!l!}\sum_{m=0}^{k-1}C^{m}_{k}\Vert \partial^{m}_{x}\partial^{l+1}_{u}g(t)\Vert_{\infty}
\Vert \partial^{k-m+1}_{x}Q(t)\Vert_{\infty}\\
&\quad +\sum_{k,l\in\setA:k+l\geq a,l\geq 1}\frac{d^{a}}{d\lambda^{a}}\frac{\lambda^{k+l}}{k!l!}\sum_{n=0}^{l-1}C^{n}_{l}
\Vert\partial^{n+1}_{u}\partial^{k}_{x}g(t)\Vert_{\infty}\Vert\partial^{l-n+1}_{u}\ln(\omega)\Vert_{\infty}\\
&\quad+\sum_{k,l\in\setA: k+l\geq a}\frac{d^{a}}{d\lambda^{a}}\frac{\lambda^{k+l}}{k!l!}\sum_{n=0}^{l}C^{n}_{l}\Vert\partial^{k}_{x}\partial^{n}_{u}g(t)\Vert_{\infty}\Vert\partial^{l-n}_{u}h\Vert_{\infty}\\
&\quad +\sum_{k,l\in\setA: k+l\geq a}\frac{d^{a}}{d\lambda^{a}}\frac{\lambda^{k+l}}{k!l!}\sum_{n=0}^{l}\sum_{m=0}^{k}C^{n}_{l}C^{m}_{k}\Vert\partial^{m}_{x}\partial^{n}_{u}g(t)\Vert_{\infty}
\Vert\partial^{k-m+1}_{x}Q(t)\Vert_{\infty}\Vert\partial^{l-n+1}_{u}\ln(\omega)\Vert_{\infty}.
\end{aligned}
\end{equation}
To bound from above the first sum on the r.h.s. of \eqref{proofst4:AnalyticalEstim} we  observe that
\begin{equation*}
\sum_{k,l\in\setA; k+l\geq a,l\geq 1}\frac{d^{a}}{d\lambda^{a}}\frac{l\lambda^{k+l}}{k!l!}\Vert\partial^{k+1}_{x}\partial^{l-1}_{u}g(t)\Vert_{\infty}=
\frac{d^{a}}{d\lambda^{a}}\sum_{k,l\in\setA;l\geq 1}\frac{l\lambda^{k+l}}{k!l!}\Vert\partial^{k+1}_{x}\partial^{l-1}_{u}g(t)\Vert_{\infty}
\end{equation*}
with
\begin{align*}
\sum_{k,l\in\setA:l\geq 1}\frac{l\lambda^{k+l}}{k!l!}\Vert\partial^{k+1}_{x}\partial^{l-1}_{u}g(t)\Vert_{\infty}&=
\sum_{k,l\in\setA:l\geq 1}\frac{\lambda^{k+l}}{k!(l-1)!}\Vert\partial^{k+1}_{x}\partial^{l-1}_{u}g(t)\Vert_{\infty}\\
&=\sum_{k,l\in\setA }\frac{\lambda^{k+l+1}}{k!l!}\Vert\partial^{k+1}_{x}\partial^{l}_{u}g(t)\Vert_{\infty}
=\lambda \Vert \partial_{x}g(t)\Vert_{\lambda,0;A}.
\end{align*}
Since
\begin{align*}
\frac{d^{a}}{d\lambda^{a}}\left(\lambda \Vert \partial_{x}g(t)\Vert_{\lambda,0;A}\right)=\sum_{r=0}^{a}C^{r}_{a}
\left(\frac{d^{r}}{d\lambda^{r}}\lambda\right)\left(\frac{d^{a-r}}{d\lambda^{a-r}}\Vert \partial_{x}g(t)\Vert_{\lambda,0;A}\right)
&=C^{a}_{a}\lambda \Vert \partial_{x}g(t)\Vert_{\lambda,a;A}+ C^{a-1}_{a}\Vert \partial_{x}g(t)\Vert_{\lambda,a-1;A}\\
&=\lambda \Vert \partial_{x}g(t)\Vert_{\lambda,a;A}+a\Vert \partial_{x}g(t)\Vert_{\lambda,a-1;A},
\end{align*}
it follows that
\begin{equation*}
\sum_{k,l\in\setA: k+l\geq a, l\geq 1}\frac{d^{a}}{d\lambda^{a}}\frac{l\lambda^{k+l}}{k!l!}\Vert\partial^{k+1}_{x}\partial^{l-1}_{u}g(t)\Vert_{\infty}
= \lambda\Vert g(t)\Vert_{\lambda,a+1;A}+a\Vert g(t)\Vert_{\lambda,a;A}.
\end{equation*}
For the second sum, we notice that
\begin{align*}
\sum_{k,l\in\setA: k\geq 1}\frac{\lambda^{k+l}}{k!l!}\sum_{m=0}^{k-1}C^{m}_{k}\Vert \partial^{k-m+1}_{x}Q(t)\Vert_{\infty}
\Vert \partial^{m}_{x}\partial^{l+1}_{u}g(t)\Vert_{\infty}&=\sum_{m,l\in\setA}\Vert \partial^{m}_{x}\partial^{l+1}_{u}g(t)\Vert_{\infty}
\left(\sum_{k=m+1}^{A}\frac{C^{m}_{k}\lambda^{k+l}}{k!l!}\Vert \partial^{k-m+1}_{x}Q(t)\Vert_{\infty}\right)\\
&=\sum_{m,l\in\setA}
\frac{\lambda^{m+l}}{m!l!}\Vert \partial^{m}_{x}\partial^{l+1}_{u}g(t)\Vert_{\infty}
\left(\sum_{k=m+1}^{A}\frac{\lambda^{k-m}}{(k-m)!}\Vert \partial^{k-m+1}_{x}Q(t)\Vert_{\infty}\right)\\
&=\sum_{m,l\in\setA}
\frac{\lambda^{m+l}}{m!l!}\Vert \partial^{m}_{x}\partial^{l+1}_{u}g(t)\Vert_{\infty}
\left(\sum_{k=1}^{A-m}\frac{\lambda^{k}}{k!}\Vert \partial^{k+1}_{x}Q(t)\Vert_{\infty}\right)\\
&=\sum_{m,l\in\setA}
\frac{\lambda^{m+l}}{m!l!}\Vert \partial^{m}_{x}\partial^{l+1}_{u}g(t)\Vert_{\infty} \Vert Q(t)\Vert_{\lambda,1;A-m}.
\end{align*}
Taking the $a$-th derivative with respect to $\lambda$, and noting that  $\sum_{m,l\in\setA}
\frac{\lambda^{m+l}}{m!l!}\Vert\partial^{m}_{x}\partial^{l+1}_{u}g(t)\Vert_{\infty}=\|\partial_u g(t) \|_{\lambda,0;A} = \|g(t)\|_{\lambda,1;A}-\|\partial_x g(t)\|_{\lambda,0;A}$ (by similar computations as   proof of Lemma \ref{lem:normprop}-(i)), we deduce that
\begin{equation*}
\frac{d^{a}}{d\lambda^{a}} \left(\sum_{k,l\in\setA:k\geq 1}\frac{\lambda^{k+l}}{k!l!}\sum_{m=0}^{k-1}C^{m}_{k}\Vert \partial^{k-m+1}_{x}Q(t)\Vert_{\infty}
\Vert \partial^{m}_{x}\partial^{l+1}_{u}g(t)\Vert_{\infty}\right)\leq \frac{d^{a}}{d\lambda^{a}}\bigg( \Vert Q(t)\Vert_{\lambda,1;A}\Vert g(t)\Vert_{\lambda,1;A}\bigg)
\end{equation*}
using also the fact that  $\frac{d^{b}}{d\lambda^{b}}\| \partial_x g(t) \|_{\lambda,0;A}\geq 0$ and $\Vert Q(t)\Vert_{\lambda,a-b+ 1;A-m}\leq \Vert Q(t)\Vert_{\lambda,a-b+1;A}$  for all $b\in \{0,\dots,a\}$.
In the same way, we obtain the  estimate
\begin{equation*}
\frac{d^{a}}{d\lambda^{a}}\left(\sum_{k,l\in\setA}\frac{\lambda^{k+l}}{k!l!}\sum_{n=0}^{l-1}C^{n}_{l}
\Vert\partial^{l-n+1}_{u}\ln(\omega)\Vert_{\infty}\Vert\partial^{n+1}_{u}\partial^{k}_{x}g(t)\Vert_{\infty}\right)\leq \frac{d^{a}}{d\lambda^{a}}\bigg(  \Vert \ln(\omega)\Vert_{\lambda,1;A}\Vert g(t)\Vert_{\lambda,1;A}  \bigg).
\end{equation*}
For the fourth sum, one can directly check that
\begin{align*}
\sum_{k,l\in\setA}\frac{\lambda^{k+l}}{k!l!}\sum_{n=0}^{l}C^{n}_{l}
\Vert\partial^{k}_{x}\partial^{n}_{u}g(t)\Vert_{\infty}\Vert\partial^{l-n}_{u}h\Vert_{\infty}
&= \sum_{k,n\in\setA}\Vert\partial^{k}_{x}\partial^{n}_{u}g(t)\Vert_{\infty}\sum_{l=n}^{A}\frac{C^{n}_{l}\lambda^{k+l}}{k!l!}
\Vert\partial^{l-n}_{u}h\Vert_{\infty}\\
&= \sum_{k,n\in\setA}\Vert\partial^{k}_{x}\partial^{n}_{u}g(t)\Vert_{\infty}\sum_{l=0}^{A}\frac{C^{n}_{l+n}\lambda^{k+l+n}}{k!(l+n)!}
\Vert\partial^{l}_{u}h\Vert_{\infty}\\
&= \sum_{k,n\in\setA}\frac{\lambda^{k+n}}{k!n!}\Vert\partial^{k}_{x}\partial^{n}_{u}g(t)\Vert_{\infty}\sum_{l=0}^{A}\frac{\lambda^{l}}{l!}
\Vert\partial^{l-n}_{u}h\Vert_{\infty}\\
&=\Vert h\Vert_{\lambda,0; A}\Vert g(t)\Vert_{\lambda,0;A},
\end{align*}
so that
\begin{equation*}
\frac{d^{a}}{d\lambda^{a}}\left(\sum_{k,l\in\setA}\frac{\lambda^{k+l}}{k!l!}\sum_{n=0}^{l}C^{n}_{l}
\Vert\partial^{k}_{x}\partial^{n}_{u}g(t)\Vert_{\infty}\Vert\partial^{l-n}_{u}h\Vert_{\infty}\right)\leq \frac{d^{a}}{d\lambda^{a}}\left(\Vert h\Vert_{\lambda,0; A}\Vert g(t)\Vert_{\lambda,0;A}\right).
\end{equation*}
Finally, since
\begin{align*}
&\sum_{k,l\in\setA}\frac{\lambda^{k+l}}{k!l!}\sum_{n=0}^{l}\sum_{m=0}^{k}C^{n}_{l}C^{m}_{k}\Vert\partial^{k-m+1}_{x}Q(t)\Vert_{\infty}\Vert\partial^{l-n+1}_{u}\ln(\omega)\Vert_{\infty}
\Vert\partial^{m}_{x}\partial^{n}_{u}g(t)\Vert_{\infty}\\
&=\sum_{m,n\in\setA}\frac{\lambda^{m+n}}{m!n!}\Vert\partial^{m}_{x}\partial^{n}_{u}g(t)\Vert_{\infty}
\left(\sum_{k=m}^{A}\frac{\lambda^{k-m}}{(k-m)!}\Vert\partial^{k-m+1}_{x}Q(t)\Vert_{\infty}\right)
\left(\sum_{l=n}^{ A}\frac{\lambda^{l-n}}{(l-n)!}\Vert\partial^{l-n+1}_{u}\ln(\omega)\Vert_{\infty}\right)\\
&\leq\Vert g(t)\Vert_{\lambda,0;A}\Vert Q(t)\Vert_{\lambda,1; A}\Vert \ln(\omega)\Vert_{\lambda,1; A},
\end{align*}
the last sum is bounded from above by
\begin{equation*}
\frac{d^{a}}{d\lambda^{a}}\bigg(\Vert g(t)\Vert_{\lambda,1; A}\left(\Vert Q(t)\Vert_{\lambda,1;A}+\Vert \ln(\omega)\Vert_{\lambda,1;A}
\right)\bigg).
\end{equation*}
Coming back to \eqref{proofst4:AnalyticalEstim}, the above estimates prove Lemma \ref{lem:AnalyticalEstim1d}.
\end{proof}
\subsubsection{Evolution and control of the time-inhomogeneous analytic norms}
Next Lemmas \ref{lemma:NormProductThreeFunctions} and \ref{lemma:NormProductTwoFunctions} are preliminaries for the bounds of the time derivative of $ \Vert g(t)\Vert_{\Hh,\lambda(t);A}$ in Proposition \ref{prop:TimeEvolutionNormSolution} below. Their proof is given in Appendix \ref{sec:LemmasProof}.
\begin{lemma}\label{lemma:NormProductThreeFunctions}
Let $f,v,w$ be functions of class  $\Cc^{\infty}$ with bounded derivatives at all order. Then, for all $\lambda>0$ and $A\in\NN$,
\begin{equation}\label{truncated:NormProductThreeFunctions}
\sum_{a\in \setA} \frac{1}{(a!)^{2}}\frac{d^{a}}{d\lambda^{a}}\left(\Vert f\Vert_{\lambda,0;A}\Vert v\Vert_{\lambda,1; A}\Vert w\Vert_{\lambda,1;A}\right)
\leq \Vert f\Vert_{\Hh,\lambda; A}\Vert v\Vert_{\tHh,\lambda; A}\Vert w\Vert_{\tHh,\lambda; A}.
\end{equation}
Suppose  moreover that  for some $\bar{\lambda}>0$,  one has $f\in \Hh(\bar{\lambda})$ and $v,w\in \tHh(\bar{\lambda})$. Then,  for all $\lambda\in [0,\bar{\lambda})$,
\begin{equation*}
\sum_{a\in\NN}\frac{1}{(a!)^{2}}\frac{d^{a}}{d\lambda^{a}}\left(\Vert f\Vert_{\lambda,0}\Vert v\Vert_{\lambda,1}\Vert w\Vert_{\lambda,1}\right)
\leq \Vert f\Vert_{\Hh,\lambda}\Vert v\Vert_{\tHh,\lambda}\Vert w\Vert_{\tHh,\lambda}.
\end{equation*}
\end{lemma}

\begin{lemma}\label{lemma:NormProductTwoFunctions}
Let $f,w$ be functions of class  $\Cc^{\infty}$ with bounded derivatives at all order.
 \begin{itemize}
\item[(i)]For all $\lambda>0$ and $A\in \NN$,  one has
\begin{equation}\label{truncated:NormProductTwoFunctionsA}
\sum_{a\in \setA}\frac{1}{(a!)^{2}}\frac{d^{a}}{d\lambda^{a}}\left(\Vert f\Vert_{\lambda,1;A}\Vert v \Vert_{\lambda,1;A}\right)\leq
 16 ( \Vert f\Vert_{\Hh,\lambda;A}\Vert v \Vert_{\tHh,\lambda;A} + \Vert f\Vert_{\tHh,\lambda;A}\Vert v  \Vert_{\Hh,\lambda;A}).
 \end{equation}
 Moreover if for some  $\bar{\lambda}>0$  we have  $f,v\in \Hh(\bar{\lambda})\cap\tHh(\bar{\lambda})$ then,  for all $\lambda\in [0,\bar{\lambda})$ \begin{equation*}
\sum_{a\in\NN}\frac{1}{(a!)^{2}}\frac{d^{a}}{d\lambda^{a}}\left(\Vert f\Vert_{\lambda,1}\Vert v \Vert_{\lambda,1}\right)\leq
 16 ( \Vert f\Vert_{\Hh,\lambda}\Vert v \Vert_{\tHh,\lambda} + \Vert f\Vert_{\tHh,\lambda}\Vert v  \Vert_{\Hh,\lambda}).
\end{equation*}
\item[(ii)]    For all $\lambda>0$ and $A\in \NN$,  one has
\begin{equation}\label{truncated:NormProductTwoFunctionsB}
\sum_{a\in\setA}\frac{1}{(a!)^{2}}\frac{d^{a}}{d\lambda^{a}}\left(\Vert f\Vert_{\lambda,1;A}\Vert v \Vert_{\lambda,1;A}\right)\leq
 4 \Vert v \Vert_{\tHh,\lambda;A} (4 \Vert f\Vert_{\Hh,\lambda;A} + \Vert f\Vert_{\tHh,\lambda;A}).
\end{equation}
Moreover for some  $\bar{\lambda}>0$, $f\in \Hh(\bar{\lambda})\cap\tHh(\bar{\lambda})$ and  $v\in \tHh(\bar{\lambda})$, for all $\lambda\in [0,\bar{\lambda})$,
\begin{equation*}
\sum_{a\in\NN}\frac{1}{(a!)^{2}}\frac{d^{a}}{d\lambda^{a}}\left(\Vert f\Vert_{\lambda,1}\Vert v \Vert_{\lambda,1}\right)\leq
 4 \Vert v \Vert_{\tHh,\lambda} (4 \Vert f\Vert_{\Hh,\lambda} + \Vert f\Vert_{\tHh,\lambda}).
 \end{equation*}
\end{itemize}
\end{lemma}

\begin{prop}\label{prop:TimeEvolutionNormSolution}  For each $A\in \NN$, the ${\cal C}^{1,\infty}$ function  $g$ solution to \eqref{eq:FPw} satisfies
\begin{align*}
\frac{d}{dt}\Vert g(t)\Vert_{\Hh,\lambda(t);A}&\leq \left(\lambda(t)+1+\lambda'(t)+4\gamma_{0}+16\Vert Q(t)\Vert_{\Hh,\lambda(t)}\right)
\Vert g(t)\Vert_{\tHh,\lambda(t);A} \\
&\quad+\left(\gamma_{1}+16\gamma_{0} +\left(\gamma_{0}+16\right)\Vert Q(t)\Vert_{\tHh,\lambda(t)}
\right)\Vert g(t)\Vert_{\Hh,\lambda(t);A},
\end{align*}
where $\gamma_{0}:=\Vert\ln(\omega)\Vert_{\tHh,\lambda_{0}}$ and
$\gamma_{1}:=\Vert h\Vert_{\Hh,\lambda_{0}}$.
\end{prop} 
\begin{proof} Differentiating in time the norm $\Vert g(t)\Vert_{\Hh,\lambda(t);A}$, we get
\begin{align*}
\frac{d}{dt}\Vert g(t)\Vert_{\Hh,\lambda(t);A}&=  \sum_{a\in \setA}\frac{1}{(a!)^{2}}
\left(\lambda'(t)\frac{d^{a+1}}{d\lambda^{a+1}}\Vert g(t)\Vert_{\lambda,0; A}\right)
+\sum_{a=0}^A\frac{1}{(a!)^{2}}\left(\frac{d}{dt}\Vert g(t)\Vert_{ A}\right)\Big{|}_{\lambda=\lambda(t)}\\
&= \lambda'(t)\sum_{a=0}^A\frac{1}{(a!)^{2}}\Vert g(t)\Vert_{\lambda(t),a+1;A}
+\sum_{a=0}^A\frac{1}{(a!)^{2}}\left(\frac{d}{dt}\Vert g(t)\Vert_{ \lambda,a;A}\right)
\Big{|}_{\lambda=\lambda(t)},
\end{align*}
Dividing both sides of the inequality in  Lemma \ref{lem:AnalyticalEstim1d} by $(a!)^{2}$ and summing the resulting expression over $a\in \setA$, it follows that
\begin{equation}
\label{proofst1:TimeEvolutionNormSolution}
\begin{aligned}
\frac{d}{dt}\Vert g(t)\Vert_{\Hh,\lambda(t);A} &\leq  \sum_{a\in \setA}\frac{\lambda'(t)+\lambda(t)}{(a!)^{2}}\Vert g(t)\Vert_{\lambda(t),a+1; A}+\sum_{a=0}^A\frac{a}{(a!)^{2}}
\Vert g(t)\Vert_{\lambda(t),a; A}\\
&\quad+\sum_{a=0}^A\frac{1}{(a!)^{2}} \frac{d^{a}}{d\lambda^{a}}\bigg(\Vert g(t)\Vert_{\lambda,0; A}\left(\Vert h\Vert_{\lambda,0; A}
+\Vert Q(t)\Vert_{\lambda,1}\Vert \ln(\omega)\Vert_{\lambda,1;A}\right)\bigg)  \Big{|}_{\lambda=\lambda(t)} \\
&\quad+\sum_{a=0}^A\frac{1}{(a!)^{2}}\frac{d^{a}}{d\lambda^{a}}\bigg(
\Vert g(t)\Vert_{\lambda,1;A}\left(\Vert Q(t)\Vert_{\lambda,1;A}+\Vert\ln(\omega)\Vert_{\lambda,1;A}
\right)\bigg) \Big{|}_{\lambda=\lambda(t)}.
\end{aligned}
\end{equation}
For the first term in \eqref{proofst1:TimeEvolutionNormSolution}, we have
\begin{align*}
\sum_{a=0}^A\frac{\lambda'(t)+\lambda(t)}{(a!)^{2}}\Vert g(t)\Vert_{\lambda(t),a+1; A}&=
\left(\lambda'(t)+\lambda(t)\right)\sum_{a=0}^A\frac{(a+1)^{2}}{((a+1)!)^{2}}\Vert g(t)\Vert_{\lambda(t),a+1; A}\\
&=\left(\lambda'(t)+\lambda(t)\right)\sum_{a=0}^A\frac{(a)^{2}}{((a)!)^{2}}\Vert g(t)\Vert_{\lambda(t),a; A},
\end{align*}
and, for the second term
\begin{equation*}
\sum_{a=0}^A\frac{a}{(a!)^{2}}
\Vert g(t)\Vert_{\lambda(t),a; A}\leq  \sum_{a\in \setA}\frac{(a)^{2}}{(a!)^{2}}
\Vert g(t)\Vert_{\lambda(t),a; A},
\end{equation*}
so that
\begin{equation}
\label{proofst2:TimeEvolutionNormSolution}
\sum_{a=0}^A\frac{\lambda'(t)+\lambda(t)}{(a!)^{2}}\Vert g(t)\Vert_{\lambda(t),a+1;A} +\sum_{a=0}^A\frac{a}{(a!)^{2}}
\Vert g(t)\Vert_{\lambda(t),a; A}\leq\left(1+\lambda'(t)+\lambda(t)\right)\Vert g(t)\Vert_{\tHh,\lambda(t); A}.
\end{equation}
For the third term, observe that on one hand
\begin{align*}
\sum_{a=0}^A\frac{1}{(a!)^{2}}\frac{d^{a}}{d\lambda^{a}}\left(\Vert g(t)\Vert_{\lambda,0; A}\Vert h\Vert_{\lambda,0; A}\right)
&=\sum_{a=0}^A\frac{1}{(a!)^{2}}\sum_{r=0}^{a}C^{r}_{a}\left(\frac{d^{r}}{d\lambda^{r}}\Vert g(t)\Vert_{\lambda,0; A }\right)
\left(\frac{d^{a-r}}{d\lambda^{a-r}}\Vert h\Vert_{\lambda,0; A}\right)\\
&=\sum_{a=0}^A\frac{1}{(a!)^{2}}\sum_{r=0}^{a}C^{r}_{a}\Vert g(t)\Vert_{\lambda,r; A}\Vert h\Vert_{\lambda,a-r; A}\\
&=\sum_{r=0}^{A}\Vert g(t)\Vert_{\lambda,r; A}\sum_{a=r}^{A}\Vert h\Vert_{\lambda,a-r; A}\frac{C^{r}_{a}}{(a!)^{2}}\\
&=\sum_{r=0}^{A}\Vert g(t)\Vert_{\lambda,r; A}\sum_{a=0}^{A}\Vert h\Vert_{\lambda,a; A}\frac{C^{r}_{a+r}}{((a+r)!)^{2}} \\
&
=\sum_{r=0}^{A}\frac{\Vert g(t)\Vert_{\lambda,r; A}}{(r!)^{2}}\sum_{a=0}^{A}\frac{\Vert h\Vert_{\lambda,a-r; A}}{(a!)^{2}}\frac{a! r!}{ (a+r))! },
\end{align*}
and since  $\frac{a!r!}{(a+r)!}\leq 1$,  for all $a,r\in\NN$, we get that
\begin{equation}
\label{proofst3:TimeEvolutionNormSolution}
\sum_{a=0}^A\frac{1}{(a!)^{2}}\frac{d^{a}}{d\lambda^{a}}\left(\Vert g(t)\Vert_{\lambda,0; A}\Vert h\Vert_{\lambda,0; A}\right)
\leq \Vert g(t)\Vert_{\Hh,\lambda; A}\Vert h\Vert_{\Hh,\lambda; A}.
\end{equation}
On the other hand, inequality \ref{truncated:NormProductThreeFunctions} provides
a bound for the remaining summand in the third  term of \eqref{proofst1:TimeEvolutionNormSolution}:
\begin{equation}
\label{proofst4:TimeEvolutionNormSolution}
  \sum_{a\in \setA}\frac{1}{(a!)^{2}} \frac{d^{a}}{d\lambda^{a}}
\bigg(\Vert g(t)\Vert_{\lambda,0; A}\Vert Q(t)\Vert_{\lambda,1; A}\Vert \ln(\omega)\Vert_{\lambda,1; A }\bigg)
%
%
\leq\Vert g(t)\Vert_{\Hh,\lambda(t);A}\Vert Q(t)\Vert_{\tHh,\lambda(t);A}
\Vert \ln(\omega)\Vert_{\tHh,\lambda(t);A}.
\end{equation}
For the fourth  term in   \eqref{proofst1:TimeEvolutionNormSolution}  we use \eqref{truncated:NormProductTwoFunctionsA} and
\eqref{truncated:NormProductTwoFunctionsB}
in order to get the estimate
\begin{equation}
\label{proofst5:TimeEvolutionNormSolution}
\begin{aligned}
&\sum_{a=0}^A\frac{1}{(a!)^{2}}\frac{d^{a}}{d\lambda^{a}}\left(
\Vert g(t)\Vert_{\lambda,1;A}\left(\Vert Q(t)\Vert_{\lambda,1;A}+\Vert\ln(\omega)\Vert_{\lambda,1;A}
\right)\right)\\
&\leq  16  \Vert g(t)\Vert_{\Hh,\lambda(t);A}\left(\Vert Q(t)\Vert_{ \tHh,\lambda(t);A}
+\Vert \ln(\omega)\Vert_{\tHh,\lambda(t);A}\right)+4 \Vert g(t)\Vert_{\tHh,\lambda(t);A}\left(4 \Vert Q(t)\Vert_{\Hh,\lambda(t);A}
+ \Vert \ln(\omega)\Vert_{\tHh,\lambda(t);A}\right).
\end{aligned}
\end{equation}
Inserting \eqref{proofst2:TimeEvolutionNormSolution}, \eqref{proofst3:TimeEvolutionNormSolution},
\eqref{proofst4:TimeEvolutionNormSolution} and \eqref{proofst5:TimeEvolutionNormSolution} in \eqref{proofst1:TimeEvolutionNormSolution},
we conclude that
\begin{align*}
\frac{d}{dt}\Vert g(t)\Vert_{\Hh,\lambda(t);A}
&\leq\left(\lambda'(t)+\lambda(t)+1\right)\Vert g(t)\Vert_{\tHh,\lambda(t);A}+\Vert g(t)\Vert_{\Hh,\lambda(t);A}\Vert h\Vert_{\Hh,\lambda(t);A}\\
&\quad+\Vert g(t)\Vert_{\Hh,\lambda(t);A}\Vert Q(t)\Vert_{\tHh,\lambda(t);A}
\Vert \ln(\omega)\Vert_{\tHh,\lambda(t);A}+16
\Vert g(t)\Vert_{\Hh,\lambda(t);A}\left(\Vert Q(t)\Vert_{\tHh,\lambda(t);A}+\Vert \ln(\omega)\Vert_{\tHh,\lambda(t);A}\right)\\
&\quad+ \Vert g(t)\Vert_{\tHh,\lambda(t);A}\left(16 \Vert Q(t)\Vert_{\Hh,\lambda(t);A}+4\Vert \ln(\omega)\Vert_{\tHh,\lambda(t);A}\right).
\end{align*}
We end the proof by using the obvious upper bounds for the truncated norms.
\end{proof}
%
%
\subsubsection{Proof of Theorem \ref{thm:Closure}}
Applying  Gronwall's lemma
to the inequality in  Proposition \ref{prop:TimeEvolutionNormSolution}, we obtain that, for all $t\in[0,T]$ and $A\in\NN$,
\begin{equation}
\begin{aligned}
\label{proofst1:Closure}
& \Vert g(t)\Vert_{\Hh,\lambda(t);A}\leq \Vert g_{0}\Vert_{\Hh,\lambda_{0}}\exp\left\{\int_{0}^{t}\left(\gamma_{1}+16 \gamma_{0}+(16+\gamma_{0})\Vert Q(\theta)\Vert_{\tHh,\lambda(\theta)}\right)\,ds\right\}\\
&+\int_{0}^{t}\left(\lambda(\theta)+1+\lambda'(\theta)+ 4\gamma_{0}+16\Vert Q(\theta)\Vert_{\Hh,\lambda(\theta)}\right)\Vert g(\theta)\Vert_{\tHh,\lambda(\theta);A}\exp\left\{
\int_{\theta}^{t}\left( \gamma_{1}+16\gamma_{0}+(16+\gamma_{0})\Vert Q(\theta')\Vert_{\tHh,\lambda(\theta')}\right)\,d\theta'\right\}\,d\theta\\
&\leq \Vert g_{0}\Vert_{\Hh,\lambda_{0}}\exp\left\{ T (\gamma_{1}+16 \gamma_{0})+(16+\gamma_{0})M_2  \right\}\\
& \quad + \exp\left\{ T (\gamma_{1}+16 \gamma_{0})+(16+\gamma_{0})M_2  \right\} (\lambda_{0}-K+4\gamma_{0}+16M_1) \int_{0}^{t} \Vert g(\theta)\Vert_{\tHh,\lambda(\theta);A}d\theta.
\end{aligned}
\end{equation}
where  in the second inequality we use the facts that $Q\in  {\cal B}^{M_1}_{\lambda_0,K,T} \cap
\widetilde{{\cal B}}^{M_2}_{\lambda_0,K,T} $ and that
 $$\lambda(t)+1+\lambda'(t)+ 4\gamma_{0}+16\Vert Q(t)\Vert_{\Hh,\lambda(t)}\leq \lambda_{0}-K+4\gamma_{0}+16\Vert Q(t)\Vert_{\Hh,\lambda(t)}\leq \lambda_{0}-K+4\gamma_{0}+16M_1$$
for all $t\in [0,T]$.  From the assumptions we can choose $K>0$ such that $K< \frac{\lambda_0}{T}-1$   and
$$K -\lambda_{0}- 4\gamma_{0}- 16M_{1} \geq 1.$$
Then we deduce with \eqref{proofst1:Closure} and the latter inequality that
\begin{align*}
 \Vert g(t)\Vert_{\Hh,\lambda(t);A} +   \int_{0}^{t} \Vert g(\theta)\Vert_{\tHh,\lambda(\theta);A}d\theta
 & \leq \Vert g(t)\Vert_{\Hh,\lambda(t);A} +   \exp\left\{ T (\gamma_{1}+16 \gamma_{0})+(16+\gamma_{0})M_2  \right\}
\int_{0}^{t} \Vert g(\theta)\Vert_{\tHh,\lambda(\theta);A}d\theta \\
 &\leq  \Vert g_{0}\Vert_{\Hh,\lambda_{0}} \exp\left\{ T (\gamma_{1}+16 \gamma_{0})+(16+\gamma_{0})M_2  \right\}.
\end{align*}
After letting $A\to\infty$ we conclude that
\begin{equation}
\label{proofst2:Closure}
\begin{aligned}
\max_{t\in[0,T]}\Vert g(t)\Vert_{\Hh,\lambda(t)}
&\leq  \Vert g_{0}\Vert_{\Hh,\lambda_{0}}  \exp\left\{ T (\gamma_{1}+16 \gamma_{0})+(16+\gamma_{0})M_2  \right\},\\
\int_{0}^{T}\Vert g(s)\Vert_{\tHh,\lambda(t)} \,dt&\leq  \Vert g_{0}\Vert_{\Hh,\lambda_{0}}  \exp\left\{ T (\gamma_{1}+16 \gamma_{0})+(16+\gamma_{0})M_2  \right\}.
\end{aligned}
\end{equation}

\subsection{Proof of Theorem \ref{thm:ExistenceNonlinear} : solving the Vlasov-Fokker-Planck equation \eqref{eq:VFP}}\label{sec:ConstructionSolutionNonlinear}
Relying upon Theorem \ref{thm:Closure}, we construct now,  by means of a Banach fixed point method, a solution to the nonlinear Vlasov-Fokker-Planck equation \eqref{eq:VFP}.

\begin{remark}
Since we are assuming in \Hw \, that $\int_{\er}\frac{u^2}{\omega(u)} du=1$, for all $\lambda\geq 0$ and $a\in \NN$ it holds that
 $$\left\|\int_{\er}\frac{u^2}{\omega(u)} \varphi(t,\cdot ,u)du\right\| _{\lambda,a}
 \leq \left\|\varphi(t,\cdot,\cdot )\right\|_{\lambda,a}$$
for any  function $\varphi:[0,T]\times \er^2 \to \er$ of class ${\cal C}^{1,\infty}$ and every $t\in [0,T]$. Therefore,  if we denote by $\Phi$ the mapping  associating to
 a function  $\varphi$  the solution $\Phi(\varphi)$ of the linear equation \eqref{eq:FPw} with  potential $\partial_{x}Q(t,x)$  given by   $$Q(t,x):=-  \int_{\er}\frac{u^2}{\omega(u)} \varphi(t,x,u)\, du \, , $$
the inclusion $$\Phi\left( {\cal B}^{M_1}_{\lambda_0,K,T} \cap
\widetilde{{\cal B}}^{M_2}_{\lambda_0,K,T} \right)  \subseteq {\cal B}^{\hat{M}}_{\lambda_0,K,T} \cap
\widetilde{{\cal B}}^{\hat{M}}_{\lambda_0,K,T} $$  holds under the
 conditions on the constants $T,\lambda_0,K, M_1,M_2$ and $\hat{M}$  established in Theorem \ref{thm:Closure}.
\end{remark}

\begin{corollary}\label{stability}
 If in addition to the assumptions of Theorem \ref{thm:Closure}, the constants $M:=M_1$ and $T>0$ satisfy
 the constraint
$$\Vert g_{0}\Vert_{\Hh,\lambda_{0}} \exp (T (\gamma_1+16\gamma_0)) \leq M \exp( -(16+\gamma_0)M),$$
then $\Phi\left( {\cal B}^{M}_{\lambda_0,K,T} \cap
\widetilde{{\cal B}}^{M}_{\lambda_0,K,T}\right)  \subseteq  {\cal B}^{M}_{\lambda_0,K,T} \cap
\widetilde{{\cal B}}^{M}_{\lambda_0,K,T} $.
\end{corollary}

\begin{proof}
Taking $M_2=M=M_1$ in  Theorem \ref{thm:Closure} we get that  $\Phi\left( {\cal B}^{M}_{\lambda_0,K,T} \cap\widetilde{{\cal B}}^{M}_{\lambda_0,K,T}\right)\subseteq   {\cal B}^{\hat{M}}_{\lambda_0,K,T} \cap
\widetilde{{\cal B}}^{\hat{M}}_{\lambda_0,K,T}$
for
$\hat{M}=\Vert g_{0}\Vert_{\Hh,\lambda_{0}}  \exp\left\{ T (\gamma_{1}+16 \gamma_{0})+(16+\gamma_{0})M  \right\}$. The additional constraint ensures that $\hat{M}\leq M$.
\end{proof}


\begin{thm}\label{thm:FixedPointMethod}
Under the assumptions of Corollary
\ref{stability} and, moreover, that
\begin{equation}\label{contraction}
M(1+\gamma_{0})\exp\left\{\left(M\gamma_{0}+\gamma_{1}\right)T\right\}<1 \,
\end{equation}
the mapping
$$\Phi: {\cal B}^{M}_{\lambda_0,K,T} \cap
\widetilde{{\cal B}}^{M}_{\lambda_0,K,T} \to   {\cal B}^{M}_{\lambda_0,K,T} \cap
\widetilde{{\cal B}}^{M}_{\lambda_0,K,T} $$
is well defined and is a contraction for the norm
$$\max\left\{\max_{t\in [0, T]} \Vert \psi(t)\Vert_{\lambda(t),0} \, , \int_{0}^{T}\Vert \psi(t)\Vert_{\lambda(t),1}\,dt \right\}.$$
If in addition to all the previous assumptions, we have
$$\max\{ \Vert g_{0}\Vert_{\Hh,\lambda_{0}} , T\Vert g_{0}\Vert_{\tHh,\lambda_{0}}\} \leq M,$$
then the (constant in time) function $g_0(t,x)=g_0(x)$ satisfies $g_0 \in {\cal B}^{M}_{\lambda_0,K,T} \cap
\widetilde{{\cal B}}^{M}_{\lambda_0,K,T} $  and a solution to the nonlinear Vlasov-Fokker-Planck equation  (VFP$\omega$)  exists in ${\cal B}^{M}_{\lambda_0,K,T} \cap
\widetilde{{\cal B}}^{M}_{\lambda_0,K,T}$.
\end{thm}

\begin{proof}
Given $f_i\in {\cal B}^{M}_{\lambda_0,K,T} \cap
\widetilde{{\cal B}}^{M}_{\lambda_0,K,T}$, $i=1,2$, we set
$P_{i}(t,x):=\int_{\er}\frac{u^2}{\omega(u)}f_{i}(t,x,u)\,du$ for $i=1,2$.  The difference $\Phi(f_{1})-\Phi(f_{2})$ satisfies
\begin{equation*}
\begin{aligned}
&\partial_{t}\left(\Phi(f_{1})-\Phi(f_{2})\right)+\left(u \partial_{x}\left(\Phi(f_{1})-\Phi(f_{2})\right)\right)-\left[\left(
\partial_{x}P_{1}-\partial_{u}\ln(\omega)\right)\partial_{u}\left(\Phi(f_{1})-\Phi(f_{2})\right)\right]
-\frac{1}{2}\partial^{2}_{u}\left(\Phi(f_{1})-\Phi(f_{2})\right)\\
&=\partial_{u}\Phi(f_{2})\left(\partial_{x}P_{1}-\partial_{x}P_{2}\right)
+\Phi(f_{2})\partial_{u}\ln(\omega)\left(\partial_{x}P_{1}-\partial_{x}P_{2}\right)
+\left(\partial_{u}\ln(\omega)\partial_{x}P_{1}+h\right)\left(\Phi(f_{1})-\Phi(f_{2})\right).
\end{aligned}
\end{equation*}
Writing $\bar{\Phi}:=\Phi(f_{1})-\Phi(f_{2})$ and $\bar{P}:=P_{1}-P_{2}$, we get
\begin{align*}
&\partial_{t}\bar{\Phi}+\left(u \partial_{x}\bar{\Phi}\right)-\left(\left(
 \partial_{x}P_{1}-\partial_{u}\ln(\omega)\right)\partial_{u}\hat{\Phi}\right)
-\frac{1}{2}\partial^{2}_{u}\bar{\Phi}\\
&= \left(\Phi(f_{2})\partial_{u}\ln(\omega)+\partial_{u}\Phi(f_{2})\right)\partial_{x}\bar{P}
+ \left(\partial_{x}P_{1}\partial_{u}\ln(\omega)+h\right)\bar{\Phi}.
\end{align*}
Then, by similar computations as in the proof of Theorem \ref{thm:Closure}, we successively obtain:\\
$\bullet$ by applying the operator $\partial^{k}_{x}\partial^{l}_{u}$,
\begin{align*}
&\partial_{t}(\partial^{k}_{x}\partial^{l}_{u}\bar{\Phi})+ u \partial_ x \left( \partial^{k}_{x}\partial^{l}_{u}\bar{\Phi}\right)-
\left( \partial_{x}P_{1} - \partial_{u}\ln(\omega)\right)\partial_u \left(\partial^{k}_{x}\partial^{l}_{u}\bar{\Phi}\right)
-\frac{1}{2}\partial^{2}_{u}(\partial^{k}_{x}\partial^{l}_{u}\bar{\Phi})\\
&=-l\partial^{k+1}_{x}\partial^{l-1}_{u}\bar{\Phi}+ \1_{\{k\geq 1\}}\sum_{m=0}^{k-1}C^{m}_{k}
 \left(\partial^{k-m+1}_{x}P_{1}\right)\partial^{m}_{x}\partial^{l+1}_{u}\bar{\Phi} -  \1_{\{l\geq 1\}} \sum_{n=0}^{l-1}C^{n}_{l}
 \left(\partial^{l-n+1}_{u}\ln(\omega)\right)\partial^{k}_{x}\partial^{n+1}_{u}\bar{\Phi} \\
&\quad +\sum_{n=0}^{l} \sum_{m=0}^{k} C^{m}_{k}C^{n}_{l}
 \partial^{k-m}_{x}\partial^{n}_{u}\Phi(f_{2}) \partial^{l-n+1}_{u}\ln \omega(u) \partial^{m+1}_{x}\bar{P}
+ \sum_{m=0}^{k}C^{m}_{k}\left(\partial^{k-m}_{x}\partial^{l+1}_{u}\Phi(f_{2})\right)\left(\partial^{m+1}_{x}\bar{P}\right)\\
&\quad+ \sum_{n=0}^{l} \sum_{m=0}^{k} C^{m}_{k}C^{n}_{l}  \left( \partial^{k-m+1}_{x}P_{1}\right)
\left(\partial^{l-n+1}_{u}\ln(\omega)\right)\left(\partial^{m}_{x}\partial^{n}_{u}\bar{\Phi}\right)  +\sum_{n=0}^{l}C^{n}_{l}
\left(\partial^{k}_{x}\partial^{n}_{u}\Phi\right)\partial^{l-n}_{u}h ;
\end{align*}
\noindent
$\bullet$ by a maximum principle, and the fact that   for all $m\in \NN$:  $\| \partial^m P_i\|_{\infty}\leq  \| \partial^m f_i\|_{\infty}$,  $i=1,2$, and $\| \partial^m \bar{P}\|_{\infty}\leq  \| \partial^m \bar{f}\|_{\infty}$  for $\bar{f}:=f_1- f_2$, we get
\begin{equation*}
\begin{aligned}
&\frac{d}{dt}\Vert\partial^{k}_{x}\partial^{l}_{u}\bar{\Phi}(t)\Vert_{\infty}\\
&\leq l\Vert\partial^{k+1}_{x}\partial^{l-1}_{u}\bar{\Phi}(t)\Vert_{\infty}+ \1_{\{k\geq 1\}} \sum_{m=0}^{k-1}C^{m}_{k}\Vert \partial^{k-m+1}_{x}f_{1}(t)\Vert_{\infty}\Vert \partial^{m}_{x}\partial^{l+1}_{u}\bar{\Phi}(t)\Vert_{\infty}\\
&\quad + \1_{\{l\geq 1\}} \sum_{n=0}^{l-1}C^{n}_{l}\Vert\partial^{l-n+1}_{u}\ln(\omega)\Vert_{\infty}\Vert\partial^{k}_{x}\partial^{n+1}_{u}\bar{\Phi}(t)\Vert_{\infty}\\
&\quad + \sum_{n=0}^{l}\sum_{m=0}^{k}C^{m}_{k}C^{n}_{l}\Vert\partial^{l-n+1}_{u}\ln(\omega)\Vert_{\infty}
\Vert\partial^{k-m}_{x}\partial^{n}_{u}\Phi(f_{2})(t)\Vert_{\infty}\Vert\partial^{m+1}_{x}\bar{f}(t)\Vert_{\infty}+ \sum_{m=0}^{k}C^{m}_{k}\Vert \partial^{k-m}_{x}\partial^{l+1}_{u}\Phi(f_{2})(t)\Vert_{\infty}\Vert\partial^{m+1}_{x}\bar{f}(t)\Vert_{\infty}\\
&\quad +\sum_{n=0}^{l}\sum_{m=0}^{k}C^{m}_{k}C^{n}_{l}\Vert\partial^{l-n+1}_{u}\ln(\omega)\Vert_{\infty} \Vert\partial^{k-m+1}_{x}f_{1}(t)\Vert_{\infty}
\Vert\partial^{m}_{x}\partial^{n}_{u}\bar{\Phi}(t)\Vert_{\infty}+\sum_{n=0}^{l}C^{n}_{l}\Vert\partial^{k}_{x}\partial^{n}_{u}\bar{\Phi}(t)\Vert_{\infty}\Vert\partial^{l-n}_{u}h\Vert_{\infty}.
\end{aligned}
\end{equation*}
$\bullet$ Replicating the computations in the proof of Lemma \ref{lem:AnalyticalEstim1d} for $a=0$,    $A=+\infty$, we then obtain
\begin{equation}
\begin{aligned}
\label{proofst2:FixedPoint}
\frac{d}{dt}\Vert\bar{\Phi}(t)\Vert_{\lambda,0}&\leq \lambda\Vert\bar{ \Phi}(t)\Vert_{\lambda,1}+
\Vert\bar{\Phi}(t)\Vert_{\lambda,1}\left(\Vert f_{1}(t)\Vert_{\lambda,1}+\Vert \ln(\omega)\Vert_{\lambda,1}\right)\\
&\quad+\Vert\bar{\Phi}(t)\Vert_{\lambda,0}\left(\Vert f_{1}(t)\Vert_{\lambda,1}\Vert \ln(\omega)\Vert_{\lambda,1}+\Vert h\Vert_{\lambda,0}\right)\\
&\quad+\Vert \bar{f}(t)\Vert_{\lambda,1}\left(\Vert\Phi(f_{2})(t)\Vert_{\lambda,1}+\Vert\Phi(f_{2})(t)\Vert_{\lambda,0}
\Vert \ln(\omega)\Vert_{\lambda,1}\right).
\end{aligned}
\end{equation}
Hence,
\begin{equation*}
\begin{aligned}
\frac{d}{dt}\Vert\bar{\Phi}(t)\Vert_{\lambda(t),0}&\leq \left(\lambda'(t)+\lambda(t)+\Vert f_{1}(t)\Vert_{\lambda,1}+\Vert \ln(\omega)\Vert_{\lambda(t),1}\right)\Vert \bar{\Phi}(t)\Vert_{\lambda(t),1}\\
&\quad+\Vert\bar{\Phi}(t)\Vert_{\lambda(t),0}\left(\Vert f_{1}(t)\Vert_{\lambda(t),1}\Vert \ln(\omega)\Vert_{\lambda(t),1}+\Vert h\Vert_{\lambda(t),0}\right)\\
&\quad+\Vert \bar{f}(t)\Vert_{\lambda(t),1}\left(\Vert\Phi(f_{2})(t)\Vert_{\lambda(t),1}+\Vert\Phi(f_{2})(t)\Vert_{\lambda(t),0}\Vert \ln(\omega)\Vert_{\lambda(t),1}\right).
\end{aligned}
\end{equation*}
Since, by our assumptions,
\begin{equation*}
\max_{t\in[0,T]}\Vert f_{1}(t)\Vert_{\lambda(t),1}\left(\leq \max_{t\in[0,T]}\Vert f_{1}(t)\Vert_{\Hh,\lambda(t)}\right)\leq M,~\mbox{ and }\max_{t\in[0,T]}\Vert\Phi(f_{2}(t))\Vert_{\lambda(t),1}\leq M,
\end{equation*}
we deduce that for all $t\in [0,T]$,
\begin{equation*}
\begin{aligned}
\label{proofst3:FixedPoint}
\frac{d}{dt}\Vert\bar{\Phi}(t)\Vert_{\lambda(t),0}&\leq \left(\lambda_0- K +M  +\gamma_0 \right)\Vert \bar{\Phi}(t)\Vert_{\lambda(t),1}+M(1+\gamma_0 )\Vert\bar{\Phi}(t)\Vert_{\lambda(t),0} +\left( M \gamma_0+ \gamma_ 1 \right)\Vert \bar{f}(t)\Vert_{\lambda(t),1}
\end{aligned}
\end{equation*}
thanks also to the upper-bounds
$\Vert\ln(\omega)\Vert_{\tHh,\lambda(t)}\leq \gamma_{0}=\Vert\ln(\omega)\Vert_{\tHh,\lambda_{0}}$, $\Vert h\Vert_{\Hh,\lambda(\theta),}\leq \gamma_{1}=\Vert h\Vert_{\Hh,\lambda_{0}}$.
It follows then by  Gronwall's inequality that
\begin{align*}
\Vert\bar{\Phi}(t)\Vert_{\lambda(t),0}\leq\exp\left\{T
\left( M \gamma_0+ \gamma_ 1 \right)\right\} \int_{0}^{t}  \left(\Vert\bar{\Phi}(\theta)\Vert_{\lambda(\theta),1}
(\lambda_0- K +M  +\gamma_0)+
\Vert \bar{f}(\theta)\Vert_{\lambda(\theta),1}M(1+\gamma_0 )
\right)d\theta.
\end{align*}
Observe that the current assumptions of Theorem \ref{thm:Closure}  ensure that we can choose $K\in \left(0,\frac{\lambda_0}{T}-1\right)$  such that
$$ K- \lambda_0 - M-\gamma_0 >1.$$
We thus get from the previous that for each $t\in [0,T]$,
\begin{align*}
\Vert\bar{\Phi}(t)\Vert_{\lambda(t),0} +\int_{0}^{t}\Vert \bar{\Phi}(\theta)\Vert_{\lambda(\theta),1}\,d\theta
&\leq  \Vert\bar{\Phi}(t)\Vert_{\lambda(t),0} +\exp\left\{T
\left( M \gamma_0+ \gamma_ 1 \right)\right\}
 \int_{0}^{t}\Vert \bar{\Phi}(\theta)\Vert_{\lambda(\theta),1}\,d\theta\\
& \leq M(1+\gamma_{0})\exp\left\{\left(M\gamma_{0}+\gamma_{1}\right)T\right\}\int_{0}^{T}
\Vert \bar{f}(t)\Vert_{\lambda(t),1}dt.
\end{align*}
In particular,
\begin{align*}
\max_{t\in [0,T]} \Vert\Phi(f_1)(t)-\Phi(f_2)(t) \Vert_{\lambda(t),0}
& \leq M(1+\gamma_{0})\exp\left\{\left(M\gamma_{0}+\gamma_{1}\right)T\right\}\int_{0}^{T}
\Vert f_1(t)-f_2(t)\Vert_{\lambda(t),1}dt.\\
\int_{0}^{T}\Vert \Phi(f_1)(\theta)-\Phi(f_2)(\theta) \Vert_{\lambda(\theta),1}\,d\theta
& \leq M(1+\gamma_{0})\exp\left\{\left(M\gamma_{0}+\gamma_{1}\right)T\right\}\int_{0}^{T}
\Vert  f_1(t)-f_2(t) \Vert_{\lambda(t),1}dt.
\end{align*}
The contractivity property is thus granted by \eqref{contraction}.
\end{proof}
\begin{proof}[Proof of Theorem \ref{thm:ExistenceNonlinear}]
Under the  assumptions on $\lambda_0,M$ and $T$,  Theorem \ref{thm:FixedPointMethod}  holds and,  moreover,  the assumptions on $f_0$ imply that
 $g_0 \in {\cal B}^{M}_{\lambda_0,K,T} \cap
\widetilde{{\cal B}}^{M}_{\lambda_0,K,T} $. Therefore,
 by   Banach's fixed point theorem  the sequence $\Phi^n(g_0)$ converges to a function $g\in {\cal B}^{M}_{\lambda_0,K,T} \cap
\widetilde{{\cal B}}^{M}_{\lambda_0,K,T}  $ which is a solution of (VFP$\omega$).
\end{proof}

\begin{remark}\label{rem:periodicity}
If $g_0(x,u)$ is $1$-periodic in $x$,  uniqueness of classical solutions to the linear equation  \eqref{eq:FPw} implies that  $\Phi^n(g_0)$  too is $1$-periodic in $x$  for each $n\in \NN$. Consequently,  so is the limit $g$.
\end{remark}
\section{The kinetic potential case}\label{sec:Kinet1d}
In this section we extend the previous results to the situation $\beta> 0$ and $\alpha=0$ (corresponding to the standard kinetic energy potential) or $\alpha=1$ (corresponding to the turbulent kinetic energy). We notice that the same proofs can be applied  also to the case $\beta<0$ by replacing in all estimates $\beta$ by $|\beta|$.

We consider the nonlinear Vlasov-Fokker-Planck equation with additional kinetic potential
\renewcommand{\theequation}{{VFP$\omega$K}}
\begin{equation}\label{eq:VFPwK}
\left\{
\begin{aligned}
&\partial_{t}g+ u\partial_{x}g - \left[
\partial_{x}P + \beta (u-\alpha V) -\partial_{u}\ln(\omega)\right]\partial_{u}g  -\frac{\sigma^2}{2}\partial^{2}_{u}g=g\left[ \partial_{x}P- \alpha\beta V\partial_{u}\ln(\omega)\right]-g\hat{h}\quad\mbox{ on }(0,T]\times\er^2,\\
&P(t,x)=-\int_{\er}\frac{u^{2}}{\omega(u)}g(t,x,u)\,du 	, \quad V(t,x)=\int_{\er}\frac{u}{\omega(u)}g(t,x,u)du\\
&g(0,x,u)=g_{0}(x,u)\mbox{ on }\er^2,
\end{aligned}\right.
\end{equation}
\renewcommand{\theequation}{\thesection.\arabic{equation}}
where
\begin{equation*}
\hat{h}(u):=\frac{\partial^{2}_{u}\omega(u)}{2\omega(u)}-|\partial_{u}\ln(\omega(u))|^{2} -\beta- \beta u  \partial_{u}( \ln \omega(u))   \, .
\end{equation*}
Through the relation $g(t,x,u)=\omega(u)f(t,x,u)$, equation \eqref{eq:VFPwK} is seen to be  equivalent to
\renewcommand{\theequation}{{VFPK}}
\begin{equation}\label{eq:VFPK}
\left\{
\begin{aligned}
&\partial_{t}f+u\partial_{x}f-\left(  \partial_{x}P +\beta(u-\alpha V) \right) \partial_{u}f - \beta f -\frac{\sigma^2}{2}\partial^{2}_{u}f=0\mbox{ on }(0,T]\times\er^2,\\
&P(t,x)=-\int_{\er} u^{2}f (t,x,u)du,\quad V(t,x)=\int_{\er} uf (t,x,u)du\\
&f(0,x,u)=f_{0}(x,u)\mbox{ on }\er^2
\end{aligned}\right.
\end{equation}
\renewcommand{\theequation}{\thesection.\arabic{equation}}
  (and to equation \eqref{eq:VFP_1D} if $f_0$  and the searched solution are periodic in $x$).
We next prove
\begin{thm}\label{thm:ExistenceNonlinearEspCond}
Let $M, T$ be positive constants and $\omega:\er\to (0,+\infty)$ be a function of class ${\cal C}^{\infty}$  satisfying \Hw\, for some $\lambda_0>0$  and moreover that $u  \partial_{u}( \ln \omega(u))\in \Hh(\lambda_{0})$. Define the finite constants
$$\gamma_{0}:= \Vert\ln(\omega)\Vert_{\tHh,\lambda_{0}}, \quad
\hat{\gamma}_{1}:= \Vert\hat{ h} \Vert_{\Hh,\lambda_{0}}  \quad  \mbox{and   } \, C_{\omega}:=  \int_{\er}\frac{|u| }{\omega(u)}du,$$
and assume that
\begin{itemize}
\item[a)]$T<\frac{(1+\beta) \lambda_0}{ 1 +4\gamma_0 +(1+\beta) (1 + \lambda_0 ) }$,
\item[b)]$M \leq \frac{(1+\beta) (K - \lambda_{0})- 4\gamma_{0}- 1}{16+\alpha\beta}$ for some $K\in( \frac{1  +4\gamma_0}{ 1+\beta} + \lambda_0  , \frac{\lambda_0}{T}  -1)\,(\not = \emptyset) $   and
\item[c)]$M(1+\gamma_{0}) (1+ T C_{\omega}  \alpha\beta ) \exp\left\{
\left( M(1+C_{\omega}  \alpha\beta)   \gamma_0+ \hat{\gamma}_1 \right)
T\right\} <1$.
\end{itemize}
Assume moreover that $f_0:\er^2 \to \er$ is a function of class ${\cal C}^{\infty}$ and that $g_0(x,u):=\omega(u)f_0(x,u)$ satisfies
\begin{itemize}
\item[d)]$\max\{ \Vert g_{0}\Vert_{\Hh,\lambda_{0}} , T\Vert g_{0}\Vert_{\tHh,\lambda_{0}}\} \leq M  $ and
\item[e)]$\Vert g_{0}\Vert_{\Hh,\lambda_{0}} \exp\left( T (\hat{\gamma}_{1}+16 \gamma_{0}+ \alpha \beta\gamma_0 C_{\omega} M )\right)<M\exp\left( -(16+\gamma_{0}) M   \right) $.
\end{itemize}
Then, equation \eqref{eq:VFPwK} has a unique smooth solution  $g\in{\cal B}^{M}_{\lambda_0,K,T}  \cap \widetilde{{\cal B}}^{M}_{\lambda_0,K,T} $. In particular, under the previous assumptions, a solution $f\in {\cal C}^{1,\infty}$ to \eqref{eq:VFPK} with initial condition $f_0$ exists.
\end{thm}

It is checked in  Appendix \ref{sec:WeightAnalycity}  that the function $u\mapsto \hat{h}(u)$ belongs to $\Hh(\lambda_{0})$ for every  $\lambda_0\in (0,\frac{1}{4})$ when $\omega(u): = c (1+u^2)^{\frac{s}{2}}$ (so that  $u  \partial_{u}( \ln \omega(u))\in \Hh(\lambda_{0})$ as required).
\begin{corollary}\label{coro:UniformMassEspCond}
Let  $f$ be the  solution to \eqref{eq:VFPK} given above  and assume that \Hu{0} holds.
Then, $f(t,x,u)$ satisfies \Hu{t} for all $t$ in $[0,T]$.
In particular, under the assumptions   of  Theorem \ref{thm:ExistenceNonlinearEspCond} and \Hu{0} a solution to \eqref{eq:CVFP}  for $\beta >0$  exists.
\end{corollary}

 \begin{remark}\label{Tsmallbeta}
Let  $f_0$   be a function of class $   {\cal C}^{\infty}$,  $C_0,\bar{\lambda}>0$ , $n,m \in \NN$ be numbers satisfying  condition \eqref{condf_0d1}  for  every $k,l\geq 0$ and assume that, moreover, for some  $\lambda_0<\min\{\bar{\lambda},\frac{1}{4}\}$ one has
$$C_0< \kappa_0'(\bar{\lambda},s):=\frac{1}{2\kappa(s)e^{\bar{\lambda}} \mu (\bar{\lambda}, m+n+1)}\frac{\ln 2}{(16+  \Vert\ln(\omega)\Vert_{\tHh,\lambda_{0}} ) } \frac{1}{(1+ C_{\omega}\alpha \beta)}$$
for $\omega$ as in Lemma   \ref{analyticw} i). Choosing  $M$  as in Remark \ref{rem:Tsmall}, one similarly checks  that
\begin{multline*}\kappa_1'(C_0,\bar{\lambda},s):=  \min \bigg\{ \frac{(1+\beta) \lambda_0}{ 1 +4\gamma_0 +(1+\beta) (1 + \lambda_0 ) + (16+\alpha\beta)M} ,  \frac{2 \bar{\lambda} \mu(\overline{\lambda},m+n+1)}{ \mu(\overline{\lambda},m+n+2)} ,1,  \\
   -\frac{\ln ( M  (1+\gamma_0)(1+ C_{\omega}\alpha \beta)  ) }{ M (1+ C_{\omega}\alpha \beta)    \gamma_0 +\hat{\gamma}_1},
\frac{\ln  2 - M (16+\gamma_0)}{\hat{\gamma}_1+16\gamma_0 + \alpha \beta C_{\omega} M }   \bigg\}>0,
\end{multline*}
 and  that the assumptions  of Theorem \ref{thm:ExistenceNonlinearEspCond} hold for  $T<\kappa_1'(C_0,\bar{\lambda},s)$ and $K=  \frac{1  +4\gamma_0+ M(16+\alpha \beta)}{ 1+\beta} + \lambda_0  $. \end{remark}

Most of the computations required  in the proofs are the same as in the previous section, so we only provide  details about the additional terms that the case $\beta>0$ requires to deal with.

\begin{proof}[Proof of Theorem \ref{thm:ExistenceNonlinearEspCond}]
Consider the linear equation  obtained by respectively replacing in \eqref{eq:VFPwK} the functions $P$ and  $V$  by fixed given functions $Q, H: [0, T]\times \er \to \er$:
\renewcommand{\theequation}{{FP$\omega$K}}
\begin{equation}\label{eq:FPwK}
\left\{
\begin{aligned}
&\partial_{t}g(t,x,u)+ u \partial_{x}g
-\left[\partial_{x}Q+\beta (u-\alpha H)-\partial_{u}\ln(\omega)\right]\partial_{u}g -\frac{\sigma^2}{2}\partial^{2}_{u}g=
g\left[ \partial_{x}Q- \alpha\beta H\partial_{u}\ln(\omega)\right]-g\hat{h},\\
&\quad\quad\mbox{ on }(0,T]\times\er^2,\\
&g(0,x,u)=g_{0}(x,u)\mbox{ on }\er^2,
\end{aligned}\right.
\end{equation}
\renewcommand{\theequation}{\thesection.\arabic{equation}}

First we notice that
$$ \partial^{k}_{x}\partial^{l}_{u}(u\partial_{u}g(t,x,u))=\sum_{n=0}^{l}C^{n}_{l}(\partial^{n}_{u}u)(\partial_{u}\partial^{k}_{x}\partial^{l-n}_{u}g(t,x,u))=u\partial^{k}_{x}\partial^{l+1}_{u}g(t,x,u)+l\partial^{k}_{x}\partial^{l}_{u}g(t,x,u).$$
Therefore,  application of the  differential operator $\partial^{k}_{x}\partial^{l}_{u}$ to  the linear equation \eqref{eq:FPwK} yields the identity

\begin{align*}
&\partial_{t} \partial^{k}_{x}\partial^{l}_{u}g+u\partial_{x} ( \partial^{k}_{x}\partial^{l}_{u}g)- \left(\partial_{x}Q(t,x)-\partial_{u}\ln \omega
+ \beta  (u-\alpha V) \right) \partial_{u} (\partial^{k}_{x}\partial^{l}_{u}g)-\frac{\sigma^2}{2}\partial^{2}_{u}(\partial^{k}_{x}\partial^{l}_{u}g)  \\
&= \1_{\{l\geq 1\}}\beta l\partial^{k}_{x}\partial^{l}_{u}g-l\partial^{k+1}_{x}\partial^{l-1}_{u}g+  \1_{\{k\geq 1\}} \sum_{m=0}^{k-1}C^{m}_{k} \partial^{k-m}_{x}(\partial_{x}Q- \alpha\beta H)
\partial_{u} (\partial^{m}_{x}\partial^{l}_{u}g)
- \1_{\{l\geq 1\}} \sum_{n=0}^{l-1}C^{n}_{l}\left(\partial^{l-n+1}_{u}\ln\omega\partial^{n+1}_{u}\partial^{k}_{x}g\right)\\
&\quad+\sum_{n=0}^{l}\sum_{m=0}^{k}C^{n}_{l}C^{m}_{k}\left(\partial^{k-m}_{x} \partial_x Q -\alpha\beta H\right) \partial^{l-n+1}_{u}\ln\omega
 \partial^{m}_{x}\partial^{n}_{u}g
+\sum_{n=0}^{l}C^{n}_{l} \partial^{k}_{x}\partial^{n}_{u}g\partial^{l-n}_{u}\hat{h}.
\end{align*}

By the maximum principle we deduce that
for all  $A\in\NN$ and $a\in\{0,\ldots,A\}$, a smooth solution $g$ to equation \eqref{eq:FPwK} must  satisfy
\begin{align*}
\frac{d}{dt}\Vert g(t) \Vert_{\lambda,a;A}&\leq\lambda(1+\beta)\Vert g(t)\Vert_{\lambda,a+1;A}
+a(1+\beta) \Vert g(t)\Vert_{\lambda,a;A}\\
&\quad+\frac{d^{a}}{d\lambda^{a}}\bigg(\Vert g(t)\Vert_{\lambda,1;A}\left(\Vert Q(t)\Vert_{\lambda,1;A}+
 \alpha \beta \Vert H (t)\Vert_{\lambda,0;A} +\Vert \ln(\omega)\Vert_{\lambda,1;A}
\right)\bigg)\\
&\quad+\frac{d^{a}}{d\lambda^{a}}\bigg(\Vert g(t)\Vert_{\lambda,0;A}\left[\Vert \hat{h}\Vert_{\lambda,0;A}
+(\Vert Q(t)\Vert_{\lambda,1;A}+\alpha \beta \Vert H (t)\Vert_{\lambda,0;A})
\Vert \ln(\omega)\Vert_{\lambda,1;A}\right]\bigg).
\end{align*}

By similar arguments as in the proofs of Lemmas \ref{lemma:NormProductTwoFunctions}   and \ref{lemma:NormProductThreeFunctions}, in  Appendix \ref{sec:LemmasProof}
we also establish
\begin{lemma}\label{lemma:NormProductTwoFunctionsEspCond}
 \begin{itemize}
\item[(i)]  Suppose that  for some  $\bar{\lambda}>0$  we have  $f\in \Hh(\bar{\lambda})$ and $v\in \tHh(\bar{\lambda})$. Then,  for all $\lambda\in [0,\bar{\lambda})$  one has

\begin{equation*}
\sum_{a\in\NN}\frac{1}{(a!)^{2}}\frac{d^{a}}{d\lambda^{a}}\left(\Vert f\Vert_{\lambda,0}\Vert v \Vert_{\lambda,1}\right)\leq
 \Vert f\Vert_{\Hh,\lambda}\Vert v \Vert_{\tHh,\lambda}.
\end{equation*}
\item[(ii)] Suppose that  for some $\bar{\lambda}>0$, $f,w \in \Hh(\bar{\lambda})$ and $v\in \tHh(\bar{\lambda})$. Then,   for all $\lambda\in [0,\bar{\lambda})$ one has
\begin{equation*}
\sum_{a\in\NN}\frac{1}{(a!)^{2}}\frac{d^{a}}{d\lambda^{a}}\left(\Vert f\Vert_{\lambda,0}\Vert w\Vert_{\lambda,0}\Vert v \Vert_{\lambda,1}\right)
\leq \Vert f\Vert_{\Hh,\lambda}\Vert w\Vert_{\Hh,\lambda}\Vert v\Vert_{\tHh,\lambda}.
\end{equation*}
\end{itemize}
\end{lemma}

Truncated version of these estimates,  combined with the already obtained ones yield:
\begin{prop}\label{prop:TimeEvolutionNormSolutionEspCond}
For each $A\in\NN$, the ${\cal C}^{1,\infty}$ function  $g$ solution to \eqref{eq:FPwK} satisfies the estimate
\begin{multline*}
\frac{d}{dt}\Vert g(t)\Vert_{\Hh,\lambda(t);A}\leq \left((1+\beta)[\lambda(t)+1+\lambda'(t)]+4\gamma_{0}+16\Vert Q(t)\Vert_{\Hh,\lambda(t)} + \alpha \beta  \Vert H(t)\Vert_{\Hh,\lambda(t)}  \right)
\Vert g(t)\Vert_{\tHh,\lambda(t);A} \\
+\left(\gamma_{1}+16\gamma_{0} +\left(\gamma_{0}+16\right)\Vert Q(t)\Vert_{\tHh,\lambda(t)}+ \alpha \beta \gamma_0 \Vert H(t)\Vert_{\Hh,\lambda(t)}
\right)\Vert g(t)\Vert_{\Hh,\lambda(t);A},
\end{multline*}
where $\gamma_{0}:=\Vert\ln(\omega)\Vert_{\tHh,\lambda_{0}}$ and
$\hat{\gamma}_{1}:=\Vert \hat{h}\Vert_{\Hh,\lambda_{0}}$.
\end{prop}
Applying  Gronwall's lemma and using the fact that
 $$
\lambda(t)+1+\lambda'(t)+ 4\gamma_{0}+16\Vert P(t)\Vert_{\Hh,\lambda(t)}+ \alpha  \beta  \Vert V(t)\Vert_{\Hh,\lambda(t)}\leq (1+\beta)[ \lambda_{0}-K] +4\gamma_{0}+(16+\alpha \beta)M_1  $$
we then obtain that, for all $t\in[0,T]$ and $A\in \NN$,
\begin{equation}
\begin{aligned}
\label{proofst1:ClosureEspCond}
& \Vert g(t)\Vert_{\Hh,\lambda(t);A} \leq \Vert g_{0}\Vert_{\Hh,\lambda_{0}}\exp\left\{ T (\hat{\gamma}_{1}+16 \gamma_{0} + \alpha \beta\gamma_0 M_1)+(16+\gamma_{0})M_2  \right\}\\
& \quad + \exp\left\{ T (\hat{\gamma}_{1}+16 \gamma_{0}+  \alpha  \beta\gamma_0 M_1 )+(16+\gamma_{0})M_2  \right\}
((1+\beta)[\lambda_{0}-K] +4\gamma_{0}+(16+\alpha  \beta) M_1) \int_{0}^{t} \Vert g(s)\Vert_{\tHh,\lambda(s);A}ds.\\
\end{aligned}
\end{equation}
 From  assumptions a)  and b)  of Theorem \ref{thm:ExistenceNonlinearEspCond} we can choose    $K>0$  such that $K< \frac{\lambda_0}{T}-1$   and
$$(1+\beta)( K - \lambda_{0}) - 4\gamma_{0}- (16+\alpha \beta) M_{1} \geq 1$$
in  which case  we obtain
\begin{equation*}
\begin{split}
 \Vert g(t)\Vert_{\Hh,\lambda(t);A} +   \int_{0}^{t} \Vert g(s)\Vert_{\tHh,\lambda(s);A}ds
 &\leq  \Vert g_{0}\Vert_{\Hh,\lambda_{0}} \exp\left\{ T (\hat{\gamma}_{1}+16 \gamma_{0}+ \alpha \beta\gamma_0 M_1 )+(16+\gamma_{0})M_2  \right\},
\end{split}
\end{equation*} and then
\begin{equation*}
\begin{split}
 \Vert g(t)\Vert_{\Hh,\lambda(t)} +   \int_{0}^{t} \Vert g(s)\Vert_{\tHh,\lambda(s)}ds
 &\leq  \Vert g_{0}\Vert_{\Hh,\lambda_{0}} \exp\left\{ T (\hat{\gamma}_{1}+16 \gamma_{0}+ \alpha \beta\gamma_0 M_1 )+(16+\gamma_{0})M_2  \right\}.
\end{split}
\end{equation*}

Therefore,  since for any function  $\varphi:[0,T]\times \er^2 \to \er$ of class ${\cal C}^{1,\infty}$   and every $t\in [0,T]$ we have
$$\left\|\int_{\er}\frac{u^2}{\omega(u)} \varphi(t,\cdot,u)du\right\|_{\lambda,a} \leq \left\|    \varphi(t,\cdot ,\cdot)\right\|_{\lambda,a}~\mbox{  and  }~ \left\|\int_{\er}\frac{|u|}{\omega(u)} \varphi(t,\cdot ,u)du\right\| _{\lambda,a}\leq C_{\omega}\left\|\varphi(t,\cdot ,\cdot ) \right\|_{\lambda,a}$$
for all $\lambda\geq 0$ and $a\in \NN$,  the mapping $\Phi$ associating with  a function $\varphi$ the solution $\Phi(\varphi)$ of equation \eqref{eq:FPwK} with the data
$$Q(t,x):=-\int_{\er}\frac{u^2}{\omega(u)} \varphi(t,x,u)du ~\mbox{ and  }~H(t,x):=\int_{\er}\frac{u}{\omega(u)}\varphi(t,x,u)du$$
satisfies the inclusion
 $$\Phi\left( {\cal B}^{M_1}_{\lambda_0,K,T} \cap
\widetilde{{\cal B}}^{M_2}_{\lambda_0,K,T} \right)\subseteq {\cal B}^{\hat{M}}_{\lambda_0,K,T} \cap
\widetilde{{\cal B}}^{\hat{M}}_{\lambda_0,K,T}$$
if $M_1, T,\lambda_0>0$ are as previously, $M_2>0$ is arbitrary and
$$\hat{M}=  \Vert g_{0}\Vert_{\Hh,\lambda_{0}} \exp\left\{ T (\hat{\gamma}_{1}+16 \gamma_{0}+ \alpha \beta\gamma_0 C_{\omega} M_1 )+(16+\gamma_{0}) M_2 \right\}.$$
In particular, one has  $\Phi\left( {\cal B}^{M}_{\lambda_0,K,T} \cap
\widetilde{{\cal B}}^{M}_{\lambda_0,K,T}\right)  \subseteq  {\cal B}^{M}_{\lambda_0,K,T} \cap
\widetilde{{\cal B}}^{M}_{\lambda_0,K,T} $  if in addition to conditions a) and b) of Theorem \ref{thm:ExistenceNonlinearEspCond}, the constants  $M>0$ and $T>0$  satisfy condition d).
Now, writing  $\bar{\Phi}:=\Phi(f_{1})-\Phi(f_{2})$,  $\bar{P}:=P_{1}-P_{2}$  and $\bar{V}:=V_1-V_2$ where
$$P_i(t,x):=-\int_{\er}\frac{u^2}{\omega(u)} f_i(t,\cdot ,u)\, du  \mbox{ and  } V_i(t,x):=   \int_{\er}\frac{u}{\omega(u)} f_i(t,\cdot,u)du \quad   i=1,2 ,$$
we have
\begin{align*}
&\partial_{t}(\partial^{k}_{x}\partial^{l}_{u}\bar{\Phi})+ u \partial_ x \left( \partial^{k}_{x}\partial^{l}_{u}\bar{\Phi}\right)-
\left( \partial_{x}Q_{1} - \partial_{u}\ln \omega (u) + \beta(u-\alpha V_1(t,x))\right)\partial_u \left(\partial^{k}_{x}\partial^{l}_{u}\bar{\Phi}\right)
-\frac{\sigma^2}{2}\partial^{2}_{u}(\partial^{k}_{x}\partial^{l}_{u}\bar{\Phi})\\
&=-l\partial^{k+1}_{x}\partial^{l-1}_{u}\bar{\Phi}+ \1_{\{k\geq 1\}}\sum_{m=0}^{k-1}C^{m}_{k}
 \partial^{k-m}_{x}\left(  \partial_x Q_{1}(t,x) - \alpha \beta V_1(t,x) \right)\partial^{m}_{x}\partial^{l+1}_{u}\bar{\Phi} -  \1_{\{l\geq 1\}} \sum_{n=0}^{l-1}C^{n}_{l}
 \left(\partial^{l-n+1}_{u}\ln(\omega)\right)\partial^{k}_{x}\partial^{n+1}_{u}\bar{\Phi} \\
&\quad +\sum_{n=0}^{l} \sum_{m=0}^{k} C^{m}_{k}C^{n}_{l}
\left(\partial^{k-m}_{x}\partial^{n}_{u}\Phi(f_{2})\right) \left(\partial^{l-n+1}_{u}\ln(\omega)\right)\partial^{m}_{x}\left(\partial_x \bar{P}-\alpha \beta \bar{V}(t,x) \right)\\
&\quad + \sum_{m=0}^{k}C^{m}_{k}\left(\partial^{k-m}_{x}\partial^{l+1}_{u}\Phi(f_{2})\right) \partial^{m}_{x}\left(\partial_x \bar{P}-\alpha \beta \bar{V}(t,x) \right) )\\
&\quad + \sum_{n=0}^{l} \sum_{m=0}^{k} C^{m}_{k}C^{n}_{l}
\partial^{k-m}_{x}\left(\partial_x P_1 -\alpha \beta V_1(t,x) \right)
 \partial^{l-n+1}_{u}\ln \omega(u)
 \partial^{m}_{x}\partial^{n}_{u}\bar{\Phi} +\sum_{n=0}^{l}C^{n}_{l}
\left(\partial^{k}_{x}\partial^{n}_{u}\bar{\Phi}\right)\partial^{l-n}_{u}\hat{h}(u) ;
\end{align*}
From this and the  maximum  principle we deduce that
$\frac{d}{dt}\Vert\partial^{k}_{x}\partial^{l}_{u}\bar{\Phi}(t)\Vert_{\infty}$ is bounded above by
\begin{align*}
& l\Vert\partial^{k+1}_{x}\partial^{l-1}_{u}\bar{\Phi}(t)\Vert_{\infty}+ \1_{\{k\geq 1\}} \sum_{m=0}^{k-1}C^{m}_{k} ( \Vert \partial^{k-m+1}_{x}f_{1}(t)\Vert_{\infty} + C_{\omega} \alpha\beta \Vert \partial^{k-m}_{x}f_{1}(t)\Vert_{\infty} ) \Vert \partial^{m}_{x}\partial^{l+1}_{u}\bar{\Phi}(t)\Vert_{\infty}\\
&\quad + \1_{\{l\geq 1\}} \sum_{n=0}^{l-1}C^{n}_{l}\Vert\partial^{l-n+1}_{u}\ln(\omega)\Vert_{\infty}\Vert\partial^{k}_{x}\partial^{n+1}_{u}\bar{\Phi}(t)\Vert_{\infty}\\
&\quad + \sum_{n=0}^{l}\sum_{m=0}^{k}C^{m}_{k}C^{n}_{l}\Vert\partial^{l-n+1}_{u}\ln(\omega)\Vert_{\infty}
\Vert\partial^{k-m}_{x}\partial^{n}_{u}\Phi(f_{2})(t)\Vert_{\infty}
( \Vert \partial^{m+1}_{x}\bar{f}(t)\Vert_{\infty} + C_{\omega}  \alpha\beta \Vert \partial^{m}_{x}\bar{f}(t)\Vert_{\infty} )\\
&\quad + \sum_{m=0}^{k}C^{m}_{k}\Vert \partial^{k-m}_{x}\partial^{l+1}_{u}\Phi(f_{2})(t)\Vert_{\infty}
( \Vert \partial^{m+1}_{x}\bar{f}(t)\Vert_{\infty} + C_{\omega}  \alpha\beta \Vert \partial^{m}_{x}\bar{f}(t)\Vert_{\infty} )\\
&\quad +\sum_{n=0}^{l}\sum_{m=0}^{k}C^{m}_{k}C^{n}_{l}\Vert\partial^{l-n+1}_{u}\ln(\omega)\Vert_{\infty}
( \Vert \partial^{k-m+1}_{x}f_1 (t)\Vert_{\infty} + C_{\omega} \alpha\beta \Vert \partial^{k-m}_{x}f_1 (t)\Vert_{\infty} )
\Vert\partial^{m}_{x}\partial^{n}_{u}\bar{\Phi}(t)\Vert_{\infty}\\
&\quad +\sum_{n=0}^{l}C^{n}_{l}\Vert\partial^{k}_{x}\partial^{n}_{u}\bar{\Phi}(t)\Vert_{\infty}\Vert\partial^{l-n}_{u}\hat{h}(u)\Vert_{\infty}.
\end{align*}
This yields
\begin{equation}
\begin{aligned}
\label{proofst2:FixedPointEspCond}
\frac{d}{dt}\Vert\bar{\Phi}(t)\Vert_{\lambda(t),0}&\leq
\left(\lambda(t)+\lambda'(t)+ \Vert f_{1}(t)\Vert_{\lambda(t),1}+ C_{\omega}  \alpha\beta \Vert f_{1}(t)\Vert_{\lambda(t),0} + \Vert \ln(\omega)\Vert_{\lambda(t),1}\right)\Vert\bar{\Phi}(t)\Vert_{\lambda(t),1} \\
&\quad+\Vert\bar{\Phi}(t)\Vert_{\lambda(t),0}\left[ \left(  \Vert f_{1}(t)\Vert_{\lambda(t),1}+ C_{\omega}  \alpha\beta \Vert f_{1}(t)\Vert_{\lambda(t),0} \right) \Vert \ln(\omega)\Vert_{\lambda(t),1}+\Vert \hat{h}\Vert_{\lambda(t),0}\right]\\
&\quad+\left( \Vert \bar{f}(t)\Vert_{\lambda(t),1}+ C_{\omega}  \alpha\beta \Vert \bar{f}(t)\Vert_{\lambda(t),0} \right) \left(\Vert\Phi(f_{2})(t)\Vert_{\lambda(t),1}+\Vert\Phi(f_{2})(t)\Vert_{\lambda(t),0}\Vert \ln(\omega)\Vert_{\lambda(t),1}\right)
\end{aligned}
\end{equation}
and therefore, for all $t\in [0,T]$,
\begin{equation}
\begin{aligned}\label{proofst3:FixedPointEspCond}
\frac{d}{dt}\Vert\bar{\Phi}(t)\Vert_{\lambda(t),0}&\leq \left( (1+\beta)( \lambda_0- K)  +(1+C_{\omega}  \alpha\beta) M  +\gamma_0 \right)\Vert \bar{\Phi}(t)\Vert_{\lambda(t),1}\\
&\quad+\Vert\bar{\Phi}(t)\Vert_{\lambda(t),0}\left( M(1+C_{\omega}  \alpha\beta)   \gamma_0+ \hat{\gamma}_1 \right)\\
&\quad+\left( \Vert \bar{f}(t)\Vert_{\lambda(t),1}+ C_{\omega}  \alpha\beta \Vert \bar{f}(t)\Vert_{\lambda(t),0} \right) M(1+\gamma_0 ).
\end{aligned}
\end{equation}
Since our assumptions allow us to choose $K\in (0,\frac{\lambda_0}{T}-1)$  such that
\begin{equation*}
(1+\beta)( K- \lambda_0)  - (1+C_{\omega}  \alpha\beta) M  - \gamma_0 >1,
\end{equation*}
we get from the previous after applying Gronwall's lemma that
\begin{equation*}
\begin{split}
\Vert\bar{\Phi}(t)\Vert_{\lambda(t),0} +\int_{0}^{t}\Vert \bar{\Phi}(s)\Vert_{\lambda(t),1}\,ds
&\leq  \Vert\bar{\Phi}(t)\Vert_{\lambda(t),0} +\exp\left\{T\left( M(1+C_{\omega}  \alpha\beta)   \gamma_0+ \hat{\gamma}_1 \right)
\right\}
 \int_{0}^{t}\Vert \bar{\Phi}(s)\Vert_{\lambda(s),1}\,ds\\
& \leq M(1+\gamma_{0})\exp\left\{
\left( M(1+C_{\omega}  \alpha\beta)   \gamma_0+ \hat{\gamma}_1 \right)
T\right\}\int_{0}^{T}\Vert \bar{f}(t)\Vert_{\lambda(t),1}+ C_{\omega}  \alpha\beta \Vert \bar{f}(t)\Vert_{\lambda(t),0} dt.\\
\end{split}
\end{equation*}
This implies that $\Phi: {\cal B}^{M}_{\lambda_0,K,T} \cap
\widetilde{{\cal B}}^{M}_{\lambda_0,K,T} \to  {\cal B}^{M}_{\lambda_0,K,T} \cap
\widetilde{{\cal B}}^{M}_{\lambda_0,K,T} $  is a contraction  for the norm
$$\max\left\{\max_{t\in [0, T]} \Vert \psi(t)\Vert_{\lambda(t),0} \, , \int_{0}^{T}\Vert \psi(t)\Vert_{\lambda(t),1}\,dt \right\}$$
under condition c) of Theorem \ref{thm:ExistenceNonlinearEspCond}, and condition d) allows us to conclude the existence of a solution starting from  $g_0$.
\end{proof}

\begin{proof}[Proof of Corollary \ref{coro:UniformMassEspCond}]
By  Lemma \ref{lem:relatedconditions}  we now obtain in the case $\alpha=1$ that, for all $(t,x) \in (0,T]\times \er$,
\begin{equation*}
\left\{
\begin{aligned}
\partial_{t}\overline{\rho}(t,x)&=- \partial_{x} V(t,x),\\
\partial_{t}(\partial_{x}V(t,x))&=- \partial_{x}(\overline{\rho}(t,x)\partial_{x}P(t,x))+\beta \partial_{x}  \left( V(t,x) \overline{\rho}(t,x)\right) ,
\end{aligned}
\right.
\end{equation*}
where $\overline{\rho}(t,x):=\rho(t,x)-1=\int_{\er} f(t,x,u)du -1 $,  $V(t,x):=\int_{\er}u f(t,x,u)\,du$ and  \\ $\partial_x  P(t,x)= - \partial_x  \int_{\er}u^{2}f(t,x,u)\,du$. Hence,
for all $\lambda>0$,
\begin{equation}
\label{proofst1:LocalPropertiesEspCond}
\left\{
\begin{aligned}
\partial_{t}\Vert\overline{\rho}(t)\Vert_{\lambda}&\leq \Vert  \partial_{x} V(t)\Vert_{\lambda},\\
\partial_{t}\Vert \partial_{x} V(t)\Vert_{\lambda}&\leq   \Vert\partial_{x}P(t)\Vert_{\lambda}\Vert\partial_{x}\overline{\rho}(t)\Vert_{\lambda}+
\Vert\partial^{2}_{x}P(t)\Vert_{\lambda}\Vert \overline{\rho}(t)\Vert_{\lambda} + \beta\left( \Vert\partial_{x} V(t)\Vert_{\lambda}\Vert \overline{\rho}(t)\Vert_{\lambda}+
\Vert  V (t)\Vert_{\lambda}\Vert \partial_{x} \overline{\rho}(t)\Vert_{\lambda}\right) .
\end{aligned}
\right.
\end{equation}

With $A(t,\lambda):=\Vert\overline{\rho}(t)\Vert_{\lambda}$ and $B(t,\lambda):=\Vert \partial_{x} V(t)\Vert_{\lambda}$ we have
\begin{equation*}
\label{proofst2:LocalPropertiesEspCond}
\left\{
\begin{aligned}
\partial_{t}A(t,\lambda)&\leq B(t,\lambda),\\
\partial_{t}B(t,\lambda)&\leq ( \Vert\partial_{x}P(t)\Vert_{\lambda}+\beta\Vert  V (t)\Vert_{\lambda})\partial_{\lambda} A(t,\lambda)+
(\Vert\partial^{2}_{x}P(t)\Vert_{\lambda}+\beta  B(t,\lambda)) A(t,\lambda).
\end{aligned}
\right.
\end{equation*}
From these inequalities,  since the terms in parentheses are bounded, the conclusion is obtained by similar arguments as in the case $\beta=0$.
If now $\alpha=0$, we obtain the equations, for all $(t,x) \in (0,T]\times \er$,
\begin{equation*}
\left\{
\begin{aligned}
\partial_{t}\overline{\rho}(t,x)&=- \partial_{x} V(t,x),\\
\partial_{t}(\partial_{x}V(t,x))&=- \partial_{x}(\overline{\rho}(t,x)\partial_{x}P(t,x))+ \beta
\partial_{x}    V(t,x)
\end{aligned}
\right.
\end{equation*}
whit the same notation as before. This yields
\begin{equation*}
\left\{
\begin{aligned}
\partial_{t}A(t,\lambda)&\leq B(t,\lambda),\\
\partial_{t}B(t,\lambda)&\leq  \Vert\partial_{x}P(t)\Vert_{\lambda} \partial_{\lambda} A(t,\lambda)+
 \Vert\partial^{2}_{x}P(t)\Vert_{\lambda} A(t,\lambda) +\beta  B(t,\lambda).
\end{aligned}
\right.
\end{equation*}
Since the remainder of the proof in the case $\beta=0$ relies on the inequality satisfied by the sum  $\mathcal{Y}(t,\lambda):=  A(t,\lambda)+bB(t,\lambda)$, by suitably modifying the constants  therein one can conclude in a similar way.
\end{proof}


\section{Local solutions for the  incompressible Langevin SDE}\label{sec:SDE}
We finally briefly state the main consequence of the previous results for  the SDE \eqref{eq:generic_incompressible_lagrangian}.

\begin{corollary}
 Let $T>0$ be a time horizon  and $p_0:\TT\times\er \to \RR_+ $  a probability density such that
\begin{itemize}
\item[$\bullet$]   $  \int_{\er}p_0(x,u)du=1$ for all $x\in\TT$,
\item[$\bullet$] $ \partial_{x}\int_{\er}u p_0(x,u)du=0$  for all $x\in\TT$,
\item[$\bullet$]  $p_0$ (or equivalently its periodic extension to $\RR^2$)  and the constant $T>0$ satisfy  the assumptions of Theorem \ref{thm:ExistenceNonlinear} (resp.  Theorem \ref{thm:ExistenceNonlinearEspCond})).
\end{itemize}
Then there exists in $[0,T]$ a solution to the stochastic differential equation  \eqref{eq:generic_incompressible_lagrangian} in the case $\sigma\neq 0$ and $\beta=0$ (resp. $\beta\in \RR$).
\end{corollary}

 \begin{proof} We deal with  the general case $\beta\in \RR$. Let $f$ be the solution to equation \eqref{eq:VFPK}   given by Theorem \ref{thm:ExistenceNonlinearEspCond} for $f_0$ equal to the periodic (in $x$) extension of $p_0$ to $\RR^2$. We know from Corollary \ref{coro:UniformMassEspCond} that $f$ is $1$- periodic, and so are also  the functions  $P_{f}(t,y):=-\int_{\er} u^{2}f (t,y,u)du$ and $ \Vv_{f}(t,y):=\int_{\er} uf (t,y,u)du$. In addition, since $P_{f}$ and $ \Vv_{f}$   have derivatives of all order in $y\in \RR$ which are bounded in $[0,T]\times  \RR$, the following  stochastic differential equation, where $W_t$ is a standard one dimensional Brownian motion independent of the random variable  $(Y_0,U_0)$,  has a pathwise unique solution  $(Y_t,U_t)$:
\begin{equation}\label{edsR}
\begin{aligned}
& dY_{t}=     \, U_t\,dt, \quad  d U_{t}= \sigma d W_{t} - \partial_{x}  P_{f}(t,Y_t)dt - \beta(U_t-\alpha \Vv_{f}(t,Y_t))dt,
\\
& law (Y_0,U_0) =  \,  f_0(y,u)\mathbf{1}_{ [0,1]}(y) dy\, du.
\end{aligned}
\end{equation}
Now, the  coefficients of \eqref{edsR}
satisfy  H\"{o}rmander's condition. Indeed, introducing the vector fields
\begin{align*}
V_{0}(x,u)&=u\partial_{x} -\left(\partial_{x}P_{f}(t,x) +\beta(u-\alpha \Vv_{f}(t,x))\right)\partial_{u},\\
V_{1}(x,u)&=\sigma \partial_{u}
\end{align*}
 the Lie bracket between $V_{0}$ and $V_{1}$ is given by
\begin{align*}
\left[V_{0},V_{1}\right](y,u)&=V_{0}\circ V_{1}(y,u)-V_{1}\circ V_{0}(y,u)\\
&=u\partial_{y}(\sigma\partial_{u})+\left(\partial_{y}P_{f} (t,y) +\beta(u-\alpha \Vv_{f}(t,y)\right)\partial_{u}(\sigma\partial_{u}) \\
& \qquad
-\sigma\partial_{u}(u\partial_{y})-\sigma\partial_{u}\left(\left(\partial_{y}P_{f} (t,y) +\beta(u-\alpha \Vv_{f}(t,y))\right)\partial_{u}\right)\\
&=-\sigma \partial_{y}+\sigma\partial_{u}\left(\partial_{y}P_{f} (t,y) +\beta(u-\alpha \Vv_{f}(t,y)\right)\partial_{u}\\
&=-\sigma \partial_{y}-\beta\sigma\partial_{u}.
\end{align*}
This shows that $Span\{V_{1},\left[V_{0},V_{1}\right]\}=\er^{2}$.
Thanks to the time regularity of the coefficients, one can use Malliavin calculus to (see \cite{Florchinger-90}) to show that $(Y_t,U_t)$ admits a $W^{\infty,1}(\er^{2})$--density $q_t$ with respect to Lebesgue measure for each $t\in [0,T]$. As a consequence the density $p_t(x,u)=\sum_{k\in \mathbb{Z}} q_t(x+k,u)$ of the random variable $(X_t,U_t)=(\left[Y_t\right] , U_t)$  in $\TT\times \RR$ is itself $W^{\infty,1}(\TT\times\er)$ so that, by classic Sobolev embeddings (see e.g. \cite{Brezis-11}) one gets  that $p_{t}\in\Cc^{\infty}(\TT\times\er)$. We further notice that $p$ is also a classical $\Cc^{1,\infty}$--solution to the linear PDE:
\begin{align*}
&\partial _{t}p +  u \partial_x p -\partial_{x}P_{f}\partial_{u}p  -\beta \left(u-  \alpha \Vv_{f} \right)\partial_{u}p- \frac{\sigma^2}{2} \partial^{2}_u p=\beta p\mbox{ on }(0,T)\times\TT\times\er,\\
&p_{t=0}=f_{0}\mbox{ on }\TT\times\er,
\end{align*}
where $P_{f}$ and $\Vv_{f}$ are considered as data. By smoothness of the functions $p$ and $f$ and the maximum principle \eqref{estim:MaxPrincipleVFPLinear} in Theorem \ref{thm:WellposedClassicLinearVFP} applied to the difference $p-f$ we deduce that $\Vert p_{t}-f(t)\Vert_{\infty}=0$ for all $t$. That is, the law   of $X_{t}=\left[Y_t\right] $ is uniform in $\TT$ for each $t$ and one has $\Vv_{f}(t,X_{t})=\EE(U_{t}|X_{t})$. This implies that $(X_{t},U_{t})$ solves the system of equations \eqref{eq:CVFP}.
\end{proof}

\appendix
\section{Appendix}

\subsection{A weight function of analytic type}\label{sec:WeightAnalycity}

  In this section, we study some properties of the weight function $\omega(u)= c(1+u^2)^{\frac{s}{2}}$
where $c, s>0$ are fixed  constants.  Notice that  if $s>3$ one has  $\int_{\er}\frac{u^{2}}{\omega(u)}du<+\infty$  and the growth conditions on $\omega$ and its derivatives required in \Hw \, are satisfied.
We will now show that the functions $u\mapsto \partial_{u}\ln(\omega(u))=\frac{s u}{(1+u^{2})}$, $u\mapsto h(u)= \frac{\partial^{2}_{u}\omega(u)}{2\omega(u)}-|\partial_{u}\ln(\omega(u))|^{2} = \frac{s  -(s+ s^{2})u^{2}}{2(1+u ^{2})^{2}}$ and  $u\mapsto \hat{h}(u) = h(u) - \beta(1+u\partial_{u}(\ln \omega(u)))$ satisfy
 $ \ln(\omega)\in \tHh(\lambda_{0})$ and $\hat{h},  h \in \Hh(\lambda_{0})$
 \Hw \,  for all $\lambda_0\in [0, 1/4)$. In particular this will prove part i) of Lemma \ref{analyticw}.

Let us first consider $\partial_{u}\ln(\omega)$. We are going to identify $\partial^{l}_{u}  \ln(\omega)$
for $l\geq 1$ with a function of the form $\frac{q_{l}(u)}{(1+u^2)^{l}}$ where $q_{l}$
 is a polynomial function of order $l$ satisfying  $q_1(u)=s u$ and, for all $l\geq 1$,
\begin{align*}
\frac{q_{l+1}(u)}{(1+u^2)^{l+1}}=
\partial_{u}\left(\frac{q_{l}(u)}{(1+u^2)^{l}}\right)=\frac{(1+u^2)\partial_{u}q_{l}(u)-2l u q_{l}(u)}{(1+u^2)^{l+1}},
\end{align*}
or,  equivalently, $q_{l+1}(u)=(1+u^2)\partial_{u}q_{l}(u)-2l u q_{l}(u)$.
We can now determine the coefficients $\{a^{(l)}_{n}\}_{0\leq n\leq l}$ such that  $q_{l}(u)= \sum_{n=0}^l a_n^{(l)}u^n$ observing that, for $l\geq 1$,
\begin{align*}
(1+u^2)\partial_{u}q_{l}(u)-2l u q_{l}(u)&=(1+u^2)\sum_{n=1}^{l}n a^{(l)}_{n}u^{n-1}-2l u\sum_{n=0}^{l} a^{(l)}_{n}u^{n}\\
&=\sum_{n=1}^{l}n a^{(l)}_{n}u^{n-1}+\sum_{n=1}^{l}n a^{(l)}_{n}u^{n+1}-2l \sum_{n=0}^{l} a^{(l)}_{n}u^{n+1}\\
&=\sum_{n=0}^{l-1}(n+1) a^{(l)}_{n+1}u^{n}+\sum_{n=2}^{l+1}(n-1) a^{(l)}_{n-1}u^{n}-2l \sum_{n=1}^{l+1} a^{(l)}_{n-1}u^{n}.
\end{align*}
Therefore, we have $a_0^{(1)}=0$, $a_1^{(1)}=s$,  $a_0^{(2)} = s , a_1^{(2)} = 0 , a_2^{(2)} =  -s$ and,  for $l\geq  2$,
\begin{equation}\label{recurrence}
\begin{aligned}
& a^{(l+1)}_{0}=a^{(l)}_{1},\quad a^{(l+1)}_{1}=2a^{(l)}_{2}-2l a^{(l)}_{0}, \\
& a^{(l+1)}_{n}=(n+1)a^{(l)}_{n+1}+(n-1)a^{(l)}_{n-1}-2l a^{(l)}_{n-1}, \mbox{ if }2\leq n\leq l-1,\\
&\mbox{ and }a^{(l+1)}_{l}=-(l+1)a^{(l)}_{l-1},\quad a_{(l+1)}^{l+1}=- l a^{(l)}_{l}.
\end{aligned}
\end{equation}
Setting $\displaystyle{ a^{(l)}:=\max_{n\in \{0,\dots , l\}}  a_n^{(l)}}$, we deduce the rough estimates: $ a^{(l+1)}\leq  4 (l+1) a^{(l)} $ for $l\geq 2$, and then:  $a^{(l)} \leq \frac{s}{4} 4^{l} l!$ for $l\geq 1$. Thus, for $l\geq 1$
\begin{equation}
|\partial^{l}_{u} (\partial_u (\ln(\omega) )|\leq \frac{\sum_{n=0}^{l+1}  a^{(l+1)}_{n} |u|^{n}}{(1+u^{2})^{l+1}}\leq  \frac{s}{4}  \sum_{n=0}^{l+1} \frac{ 4^{l+1} (l+1)! |u|^{n}}{{(1+u^{2})^{l+1}}}\leq  s 4^{l}(l+2)!.
\end{equation}
Consequently, by Lemma \ref{lem:NormCriterion} we have $\partial_u ( \ln(\omega) )\in {\cal H}(\lambda)$ for all $\lambda \in [0, 1/4)$, and from Lemma \ref{lem:normprop}-(i) we conclude that
$ \ln(\omega) \in \widetilde{{\cal H}}(\lambda)$ for all $\lambda \in [0, 1/4)$.

As for the function $h$, it is similarly checked in this case that for all $l\geq 0$
\begin{align*}
\partial^{l}_{u}h(u)=\frac{r_{l+2}(u)}{2(1+u^2)^{l+2}},
\end{align*}
where $r_{l}$ is a polynomial function of order $l$, defined for $l\geq 2$. The coefficients
 $\{b^{(l)}_{n}\}_{0\leq n\leq l}$  such that $r_l(u)=\sum_{n=0}^{l} b^{(l)}_{n} u^n$ satisfy  $b^{(2)}_{0}=s/2$, $b^{(2)}_{1}=0$, $b^{(2)}_{2}=-(s+s^2)/2$ and moreover,
for all $l\geq 2$, the  recurrence relations \eqref{recurrence}
with  $a^{(l)}_n$ replaced by $b^{(l)}_{n}$. It follows in a similar way as before that
\begin{align*}
|\partial^{l}_{u}h(u)|\leq \frac{s+s^2}{4} 4^l (l+3)!
\end{align*}
and we conclude as well by Lemma \ref{lem:NormCriterion} that $h\in{\cal H}(\lambda)$ for all  $\lambda \in [0, 1/4)$.

Finally,  we have $u\partial_{u}( \ln \omega(u))= s- \frac{s}{1+u^2}$ so that we only need to check that the function $ \frac{s}{1+u^2}$   belongs to  the space ${\cal H}(\lambda)$ for $\lambda \in [0, 1/4)$.
Plainly,  for each $l\geq 0$,   $\partial^{l}_{u}\left( \frac{s}{1+u^2}\right)  = \frac{j_{l}(u)}{(1+u^2)^{l+1}}$ for some polynomial  $j_{l}$  of order $l$. This yields a recurrence relation for the coefficients that  only differs from  \eqref{recurrence} in  that the factors $-2l $ are replaced by $-2(l+1)$. The conclusion thus follows as previously.

\medskip

\medskip

We end this technical section verifying claim  ii) of Lemma  \ref{analyticw}. Since $ \partial_u^j \omega =0$  for $j>s$ and $| \partial_u^j \omega|\leq \kappa(s) (1+u^2)^{\frac{s}{2}} $ for $j\leq s$, we have  $ \partial_u^l \partial_x ^k ( f_0(x,u) \omega(u)) = \sum_{j=0}^{s\wedge l }   \frac{l!}{(l-j)! j!}  \partial_u^j \omega(u)  \partial_u^{l-j}  (\partial_x ^k f_0(x,u) )$ and we deduce from the assumptions  that
$$ \Vert \partial^{l}_x   \partial^{k}_u g_0 \Vert_{\infty}\leq  \kappa(s) \frac{C_0(k+m)!}{\overline{\lambda}^{k+l}}    \sum_{j=0}^l   \frac{l! (l-j+n) !}{(l-j)! }  \frac{ {\overline{\lambda}^j} }{j!}\leq  \kappa(s) \frac{C_0(k+m)! (l+n)!}{\overline{\lambda}^{k+l}}  e^{ {\overline{\lambda}}}$$
using also the fact that $ \frac{l! (l-j+n) !}{(l-j)! } = l! (l-j +n)\cdots (l-j+1) \leq (l+n)!$ for all $j\leq l$.
The last assertion of Lemma  \ref{analyticw} is an immediate consequence of the previous and of the proof of Lemma  \ref{lem:NormCriterion}.

\subsection{A maximum principle for kinetic Fokker-Planck equations}\label{sec:WellPosedLinearVFP}
We next give for completeness a brief proof of the version of the maximum principle  that has been used throughout.
\begin{thm}\label{thm:WellposedClassicLinearVFP}  Let $d\geq 1$ and $\sigma \in \er$. Consider  bounded functions $f_{0}:\er^d\times\er^d\to\er $ and  $F, c: [0,T]\times \er^d \times \er^d\to \er $     and  a function $ \phi: [0,T]\times \er^d \times \er^d\to \er^d$  that grows linearly in $(x,u)$ uniformly in $t\in [0, T]$. Assume moreover that all these functions are of class $\Cc^{0,\infty}$  and have bounded derivatives of all order.  Then, there exists a unique solution $f$  of class  $\Cc^{1,\infty}$ to the linear Fokker-Planck equation  in $Q_T:=[0,T]\times\er^{2d}$ :
\begin{equation*}
\left\{
\begin{aligned}
&\partial_{t}f(t,x,u)+ u\cdot  \nabla_{x}f(t,x,u) -\phi(t,x,u)\cdot \nabla_{u}f(t,x,u)-
\frac{\sigma^2}{2}\triangle_{u}f(t,x,u)+c(t,x,u)f(t,x,u)=F(t,x,u),\\
&f(0,x,u)=f_{0}(x,u)\mbox{ on }\er^{2d}
\end{aligned}
\right.
\end{equation*}
which is bounded in that domain.  Moreover,  the function $t\mapsto  \Vert f(t)\Vert_{\infty}$ is absolutely continuous, and  for almost every $t\in [0,T]$ one has
\begin{equation}
\label{estim:MaxPrincipleVFPLinear}
\frac{d}{dt}\Vert f(t)\Vert_{\infty}\leq \Vert c(t)\Vert_{\infty}\Vert f(t)\Vert_{\infty}+\Vert F(t)\Vert_{\infty}.
\end{equation}
\end{thm}
\begin{proof}  In the case $\sigma \neq 0$,
existence of a bounded solution of class  $\Cc^{1,\infty}$ can be obtained by probabilistic methods,  considering the unique pathwise solution $(X_s^{t,x,u},U_s^{t,x,u})_{t\leq s\leq T} $ to the stochastic differential equation in $\er^d\times \er^d$:
\begin{equation*}
\begin{split}
X_{s}^{t,x,u} & =x -\int_{t}^{s} U_{r}^{t,x,u}dr\\
U_{s}^{t,x,u} & =u +\int_{t}^s \phi(T-r,X_r^{t,x,u},U_r^{t,x,u}) dr+\sigma (W_s- W_t)\\
\end{split}
\end{equation*}
where $W$ is a standard $d$-dimensional Brownian motion defined in some filtered probability space.  Following Friedman \cite{Friedman-06} p.124, one shows that
\begin{align*}
(t,x,u)\mapsto &f(t,x,u):=  \EE\left[f_0(X^{T-t,x,u}_{T},U^{T-t;x,u}_{T})\exp\left\{\int_{T-t}^T c(T-\theta,X^{T-t,x,u}_{\theta},U^{T-t,x,u}_{\theta})\,d\theta\right\}\right]\\
&+\EE\left[\int_{T-t}^{T}F(T-\theta,X^{T-t,x,u}_{\theta},U^{T-t,x,u}_{\theta})\exp\left\{\int_{T-t}^{\theta} c(T-s,X^{T-t,x,u}_{s},U^{T-t,x,u}_{s})ds\right\}d\theta\right]
\end{align*}
is a solution to the Fokker-Planck equation.  Moreover, using It\^o' s formula one shows that any bounded solution has the previous Feynman-Kac representation and is therefore unique  because of  uniqueness in law for the previous SDE. Furthermore, the Jacobian matrix of the flow $(x,u)\mapsto(X_{s}^{t,x,u},U_{s}^{t,x,u})$ satisfies a linear matrix  ODE with bounded coefficient given for each $r\in [0,T]$ by the  Jacobian matrix of the function  $(x,u)\mapsto  (u, \phi(T-r,x,u))$. This implies (by Gronwall's lemma) that $(x,u)\mapsto(X_{s}^{t,x,u},U_{s}^{t,x,u})$ has bounded derivatives of first order, and then of all order by applying inductively a similar argument. Taking derivatives under the expectation sign in the above representation and using moreover  the regularity of
$\rho_0,\phi,c,F$, one then
 deduces   that  $f(t,x,u)$  has bounded derivatives of all order in $(x,u)$.  Now  set $\bar{f}(t,x,u):=f(T-t,x,u)$ and apply  It\^o's formula to get
$$\bar{f}(s, X_{s}^{t,x,u},U_{s}^{t,x,u}) =\bar{f}(t,x,u)+\int_{t}^s  ( c f-F)(T-r, X_{r}^{t,x,u},
U_{r}^{t,x,u} )  dr  + \sigma \int_{t}^s  \nabla  \bar{f}(r, X_{r}^{t,x,u}, U_{r}^{t,x,u} )dW_r$$
for all $t\leq s\leq T$. Taking expectations we deduce that
$$\mathbb{E} f(T-s, X_{s}^{t,x,u},U_{s}^{t,x,u}) = f(T-t,x,u) +\int_{t}^{s} \mathbb{E}
\left[ ( cf -F) (T-r, X_{r}^{t,x,u},
U_{r}^{t,x,u} )  )\right]dr,$$
which implies (taking $\theta=T-s\leq \theta' =T-t$) that
\begin{equation}\label{abscont1}
 \Vert f(\theta' )\Vert_{\infty} - \Vert f(\theta)\Vert_{\infty} \leq    \int_{\theta}^{\theta'}  \Vert c(r)\Vert_{\infty}\Vert f(r)\Vert_{\infty}+\Vert F(r)\Vert_{\infty}\,  dr
 \end{equation}
for all $0\leq \theta  \leq \theta' \leq T$. Notice now that $\bar{f}$ defined above  is a classic solution of the equation
\begin{equation*}
\left\{
\begin{aligned}
&\partial_{t}\bar{f}(t,x,u)- u\cdot  \nabla_{x}\bar{f}(t,x,u) +\bar{\phi}(t,x,u)\cdot \nabla_{u} \bar{f}(t,x,u)-
\frac{\sigma^2}{2}\triangle_{u} \bar{f}(t,x,u)- \bar{c}(t,x,u) \bar{f}(t,x,u)=  \hat{F}(t,x,u),\\
&\bar{f}(0,x,u)=f(T,x,u)\mbox{ on }\er^{2d}
\end{aligned}
\right.
\end{equation*}
where $\bar{\phi}(t,x,u)=\phi(T-t,x,u)$, $\bar{c}(t,x,u)=c(T-t,x,u)$ and  $ \hat{F}(t,x,u)=-F(T-t,x,u) -\sigma^2\triangle_{u} \bar{f}(t,x,u)  $ is a bounded function. We can thus apply  inequality  \eqref{abscont1} to the function $\bar{f}$ and deduce that
\begin{equation*}
 -  \int_{\theta}^{\theta'}  \Vert c(r)\Vert_{\infty}\Vert f(r)\Vert_{\infty}+\Vert F(r)\Vert_{\infty} + \sigma^2 \Vert \triangle_u f(r)\Vert_{\infty} \,  dr \leq  \Vert f(\theta' )\Vert_{\infty} - \Vert f(\theta)\Vert_{\infty}
   \end{equation*}
for all $0\leq \theta  \leq \theta' \leq T$. This inequality  and  \eqref{abscont1} imply that  $t\mapsto  \Vert f(t)\Vert_{\infty}$ is absolutely continuous, and the  upper bound   \eqref{abscont1} then yields the asserted  bound  \eqref{estim:MaxPrincipleVFPLinear}  on the   a.e. derivative.

\medskip

 Finally,  the same arguments go through when considering  the corresponding ordinary differential equation  obtained when taking the limit $\sigma\rightarrow 0$.

\end{proof}

\subsection{Proofs of Lemmas \ref{lemma:NormProductThreeFunctions}, \ref{lemma:NormProductTwoFunctions} and \ref{lemma:NormProductTwoFunctionsEspCond}}\label{sec:LemmasProof}

We provide here  the proofs of  Lemmas \ref{lemma:NormProductThreeFunctions}, \ref{lemma:NormProductTwoFunctions} and \ref{lemma:NormProductTwoFunctionsEspCond} following  arguments of \cite{JabNou-11}. Their truncated versions used in the proof of Proposition \ref{prop:TimeEvolutionNormSolution}  are obtained in  a similar  way, namely replacing in the next proofs the norms $\|\cdot\|_{\lambda,a}$ by their truncated versions  $\|\cdot\|_{\lambda,a, A}$ for each $A\in\NN$, and the sums over $\NN$ by sums over the set $\{0,\dots,A\}$.

\begin{proof}[Proof of Lemma \ref{lemma:NormProductThreeFunctions}] By definition, we have
\begin{align*}
\sum_{a\in\NN}\frac{1}{(a!)^{2}}\frac{d^{a}}{d\lambda^{a}}\left(\Vert f\Vert_{\lambda,0}\Vert  v\Vert_{\lambda,1}\Vert w\Vert_{\lambda,1}\right)=
\sum_{a\in\NN}\frac{1}{(a!)^{2}}\sum_{r=0}^{a}C^{r}_{a}
\frac{d^{r}}{d\lambda^{r}}\Vert f\Vert_{\lambda,0}\frac{d^{a-r}}{d\lambda^{a-r}}\left(\Vert v \Vert_{\lambda,1}\Vert w\Vert_{\lambda,1} \right).
\end{align*}
Then we see that
\begin{align*}
&\sum_{a\in\NN}\frac{1}{(a!)^{2}}\sum_{r=0}^{a}C^{r}_{a}\frac{d^{r}}{d\lambda^{r}}\Vert f\Vert_{\lambda,0}
\frac{d^{a-r}}{d\lambda^{a-r}}\left(\Vert v \Vert_{\lambda,1}\Vert  w\Vert_{\lambda,1}\right)\\
&=\sum_{r\in\NN}\Vert f\Vert_{\lambda,r}\sum_{a=r}^{+\infty}\frac{C^{r}_{a}}{(a!)^{2}}
\frac{d^{a-r}}{d\lambda^{a-r}}\left(\Vert  v \Vert_{\lambda,1}\Vert  w\Vert_{\lambda,1}\right)\,\,\,\,(\mbox{since}\,\frac{d^{p}}{d\lambda^{p}}\Vert\psi\Vert_{\lambda,0}=\Vert\psi\Vert_{\lambda,p}\,\mbox{by definition})\\
&=\sum_{r\in\NN}\Vert f\Vert_{\lambda,r}\sum_{a=r}^{+\infty}\frac{C^{r}_{a}}{(a!)^{2}}
\sum_{q=0}^{a-r}C^{q}_{a-r}\left(\Vert v\Vert_{\lambda,q+1}\Vert  w\Vert_{\lambda,a-q-r+1}\right)\\
&=\sum_{r\in\NN}\Vert f\Vert_{\lambda,r}\sum_{a=0}^{+\infty}
\sum_{q=0}^{a}\frac{C^{r}_{a+r}C^{q}_{a}}{((a+r)!)^{2}}\left(\Vert v \Vert_{\lambda,q+1}\Vert  w\Vert_{\lambda,a-q+1}\right)\,\,\,\,(\mbox{by a change of variables})\\
&=\sum_{r\in\NN}\Vert f\Vert_{\lambda,r}
\sum_{q=0}^{+\infty}\left(\Vert v \Vert_{\lambda,q+1}\sum_{a=q}^{+\infty}\Vert  w\Vert_{\lambda,a-q+1}\frac{C^{r}_{a+r}C^{q}_{a}}{((a+r)!)^{2}}\right)\\
&=\sum_{r\in\NN}\Vert f\Vert_{\lambda,r}
\sum_{q=0}^{+\infty}\left(\Vert v \Vert_{\lambda,q+1}\sum_{a=0}^{+\infty}\Vert  w\Vert_{\lambda,a+1}\frac{C^{r}_{a+r+q}C^{q}_{a+q}}{((a+r+q)!)^{2}}\right).
\end{align*}
Thus,
\begin{align*}
&\sum_{a\in\NN}\frac{1}{(a!)^{2}}\frac{d^{a}}{d\lambda^{a}}\left(\Vert f\Vert_{\lambda,0}\Vert  v \Vert_{\lambda,1}\Vert w\Vert_{\lambda,1}\right)\\
&=\sum_{r\in\NN}\frac{\Vert f\Vert_{\lambda,r}}{(r!)^{2}}
\sum_{q=0}^{+\infty}\frac{(q+1)^{2}}{((q+1)!)^{2}}\Vert v \Vert_{\lambda,q+1}\sum_{a=0}^{+\infty}
\frac{(a+1)^{2}}{((a+1)!)^{2}}\Vert  w\Vert_{\lambda,a+1}\frac{(r!)^{2}((q+1)!)^{2}((a+1)!)^{2}C^{r}_{a+r+q}C^{q}_{a+q}}{(a+1)^{2}(q+1)^{2}((a+r+q)!)^{2}},
\end{align*}
where
\begin{equation*}
\frac{(r!)^{2}((q+1)!)^{2}((a+1)!)^{2}C^{r}_{a+r+q}C^{q}_{a+q}}{(a+1)^{2}(q+1)^{2}((a+r+q)!)^{2}}
=\frac{a! q! r!}{(a+q+r)!}.
\end{equation*}
The claim follows since $ \frac{a! q! r!}{(a+q+r)!}\leq 1,\,\,\forall\,a,q,r\in\NN.$

\end{proof}
\begin{proof}[Proof of Lemma  \ref{lemma:NormProductTwoFunctions}] One has
\begin{align*}
\sum_{a\in\NN}\frac{1}{(a!)^{2}}\frac{d^{a}}{d\lambda^{a}}\left(\Vert f\Vert_{\lambda,1}\Vert v \Vert_{\lambda,1}\right)&=
\sum_{a\in\NN}\frac{1}{(a!)^{2}}\left(\sum_{r=0}^{a}
C^{r}_{a}\frac{d^{r}}{d\lambda^{r}}\Vert f\Vert_{\lambda,1}\frac{d^{a-r}}{d\lambda^{a-r}} \Vert v \Vert_{\lambda,1}\right)\\
&=\sum_{a\in\NN}\frac{1}{(a!)^{2}}\left(\sum_{r=0}^{a}
C^{r}_{a}\Vert f\Vert_{\lambda,r+1} \Vert v \Vert_{\lambda,a-r+1}\right)\,\,\,\,(\mbox{since}\,\frac{d^{p}}{d\lambda^{p}}\Vert\psi\Vert_{\lambda,0}=\Vert\psi\Vert_{\lambda,p}\,\mbox{by definition})\\
&=\sum_{r\in\NN}\Vert f\Vert_{\lambda,r+1}\left(\sum_{a=r}^{+\infty}
\frac{C^{r}_{a}}{(a!)^{2}}\Vert v \Vert_{\lambda,a-r+1}\right)\\
&=\sum_{r\in\NN}\Vert f\Vert_{\lambda,r+1}\left(\sum_{a=0}^{+\infty}
\frac{C^{r}_{a+r}}{((a+r)!)^{2}}\Vert v \Vert_{\lambda,a+1}\right)\\
&=\sum_{r\in\NN}\frac{\Vert f\Vert_{\lambda,r+1}}{((r+1)!)^{2}}\sum_{a\in\NN}\frac{\Vert v \Vert_{\lambda,a+1}}{((a+1)!)^{2}}
\left(\frac{C^{r}_{a+r}((a+1)!)^{2}((r+1)!)^{2}}{((a+r)!)^{2}}\right),
\end{align*}
where
\begin{align*}
\frac{C^{r}_{a+r}((a+1)!)^{2}((r+1)!)^{2}}{((a+r)!)^{2}}&=\frac{(a+1)(r+1)(a+1)!(r+1)!}{(a+r)!}.
\end{align*}
As in \cite{JabNou-11},  we observe that when $a\geq 2$  and $r\geq 2$ one has
\begin{equation*}
\frac{(r+1)(a+1)(a+1)!(r+1)!}{(a+r)!}\leq 24.
\end{equation*}
Indeed, for $r\geq 2$ and $a\geq 3$,
\begin{align*}
\frac{(r+1)(a+1)(a+1)!(r+1)!}{(a+r)!}&=\frac{1\times \cdots \times (r+1) \times (r+1)}{1\times \cdots \times r \times (r+1) \times (r+2) }
\times1 \times 2  \times 3  \times 4 \times\frac{5 \times \cdots  \times a  \times (a+1) \times (a+1)}{  (r+3) \cdots \times (a+r-1) \times (a+r)}\,\
\\
&\leq 4! \times \frac{5  \times \cdots \times (a+1) \times (a+1)}
{(r+3) \times \cdots \times (a+r-1) \times (a+r)}\\
%
&\leq 24,
\end{align*}
where $a\geq 3$ was used in the first expansion  and $r\geq 2$ in the  second inequality.
If $r\geq 2$ and $a=2$ then
\begin{equation*}
\frac{(r+1)(a+1)(a+1)!(r+1)!}{(a+r)!}=18\times\frac{(r+1)(r+1)!}{(r+2)!}\leq 18.
\end{equation*}
If $a\leq 1$  or $r\leq 1$  we have to separate the corresponding  terms in the estimation.
Then, we get
\begin{align*}
&\sum_{a\in\NN}\frac{1}{(a!)^{2}}\sum_{r=0}^{a}C^{r}_{a}
\Vert f\Vert_{\lambda,r+1}\Vert v \Vert_{\lambda,a-r+1}= \sum_{r\in\NN}\frac{\Vert f\Vert_{\lambda,r+1}}{((r+1)!)^{2}}\sum_{a=0}^{+\infty}\frac{\Vert v \Vert_{\lambda,a+1}}{((a+1)!)^{2}}
\left( \frac{(r+1)(a+1)(a+1)!(r+1)!}{(a+r)!}\right)\\
& = \Vert  v \Vert_{\lambda,1}\sum_{r\in\NN}\frac{\Vert f\Vert_{\lambda,r+1}}{((r+1)!)^{2}}(r+1)^2 +
 4   \Vert v \Vert_{\lambda,2}\sum_{r\in \NN}\frac{\Vert f\Vert_{\lambda,r+1}}{((r+1)!)^{2}}( r+1)\\
&\quad  +\Vert f\Vert_{\lambda,1}\sum_{a\geq 2}\frac{\Vert v \Vert_{\lambda,a+1}}{((a+1)!)^{2}}
(a+1)^2 +
 4 \Vert f\Vert_{\lambda,2}\sum_{a\geq 2}\frac{\Vert v \Vert_{\lambda,a+1}}{((a+1)!)^{2}} (a+1) \\
&\quad+24 \sum_{r\geq 2}\frac{\Vert f\Vert_{\lambda,r+1}}{((r+1)!)^{2}}
\sum_{a\geq 2}\frac{\Vert v \Vert_{\lambda,a+1}}{((a+1)!)^{2}}\\
&\leq  \left(\Vert  v \Vert_{\lambda,1}+4 \Vert  v \Vert_{\lambda,2}\right)\Vert f\Vert_{\tHh,\lambda}+\left(\Vert f\Vert_{\lambda,1}+4\Vert f\Vert_{\lambda,2}\right)  \Vert v \Vert_{\tHh,\lambda }
 +24 \sum_{r\geq 2}\frac{\Vert f\Vert_{\lambda,r+1}}{((r+1)!)^{2}}
\sum_{a\geq 2}\frac{\Vert v \Vert_{\lambda,a+1}}{((a+1)!)^{2}}\\
\end{align*}
Since  $
\sum_{a\geq 3}\frac{\Vert v \Vert_{\lambda,a}}{(a!)^{2}}\leq
\Vert v\Vert_{\tHh,\lambda}\wedge \Vert v\Vert_{\Hh,\lambda} $ and  $ \sum_{r\geq 3}\frac{\Vert f \Vert_{\lambda,r}}{(r!)^{2}}\leq  \Vert f \Vert_{\tHh,\lambda}  \wedge \Vert f\Vert_{\Hh,\lambda}  $, the latter expression can be  bounded above by
\begin{multline*}
 \left(\Vert  v \Vert_{\lambda,1}+4 \Vert  v \Vert_{\lambda,2} + 12 \sum_{a\geq 3}\frac{\Vert v \Vert_{\lambda,a}}{(a!)^{2}}
\right)\Vert f\Vert_{\tHh,\lambda}+\left(\Vert f\Vert_{\lambda,1}+4\Vert f\Vert_{\lambda,2}  +12 \sum_{r\geq 3}\frac{\Vert f\Vert_{\lambda,r}}{(r!)^{2}}  \right)  \Vert v \Vert_{\tHh,\lambda } \\
\leq 16 ( \Vert f \Vert_{\tHh,\lambda}\Vert  v \Vert_{\Hh,\lambda} + \Vert  f \Vert_{\Hh,\lambda}\Vert  v \Vert_{\tHh,\lambda})
\end{multline*}
and with the bound $\Vert  w \Vert_{\lambda,1}+4 \Vert  w \Vert_{\lambda,2} + 12 \sum_{a\geq 3}\frac{\Vert w \Vert_{\lambda,a}}{(a!)^{2}}\leq 16 \Vert  w \Vert_{\Hh,\lambda}$  for $w=f,v$
we establish (i). Using the bound $\Vert  v\Vert_{\lambda,1}+4 \Vert  v\Vert_{\lambda,2} + 12 \sum_{a\geq 3}\frac{\Vert v\Vert_{\lambda,a}}{(a!)^{2}}\leq 4\Vert  v\Vert_{\tHh,\lambda}$ we alternatively obtain (ii).
\end{proof}
\begin{proof}[Proof of Lemma \ref{lemma:NormProductTwoFunctionsEspCond}]
Replicating the first computation in the proof of Lemma \ref{lemma:NormProductTwoFunctions},
we obtain
\begin{equation*}
\sum_{a\in\NN}\frac{1}{(a!)^{2}}\frac{d^{a}}{d\lambda^{a}}\left(\Vert f\Vert_{\lambda,0}\Vert v \Vert_{\lambda,1}\right)=
\sum_{r\in\NN}\frac{\Vert f\Vert_{\lambda,r}}{((r)!)^{2}}\sum_{a\in\NN}\frac{(a+1)^{2}}{((a+1)!)^{2}}\Vert v \Vert_{\lambda,a+1}
\left(\frac{C^{r}_{a+r}((a+1)!)^{2}((r)!)^{2}}{((a+r)!)^{2}}\right).
\end{equation*}
Since, for all $a,r\in \NN$,
\begin{equation*}
\frac{C^{r}_{a+r}((a+1)!)^{2}((r)!)^{2}}{((a+r)!)^{2}}=\frac{a!r!}{(a+r)!}\leq 1,
\end{equation*}
we deduce $(i)$. In the same way, $(ii)$ is obtained by replicating the first computation in the proof of Lemma \ref{lemma:NormProductThreeFunctions}
in order to get
\begin{equation*}
\sum_{a\in\NN}\frac{1}{(a!)^{2}}\frac{d^{a}}{d\lambda^{a}}\left(\Vert f\Vert_{\lambda,0}\Vert  v\Vert_{\lambda,1}\Vert w\Vert_{\lambda,0}\right)
=\sum_{r\in\NN}\frac{\Vert f\Vert_{\lambda,r}}{(r!)^{2}}
\sum_{q=0}^{+\infty}\frac{(q+1)^{2}}{((q+1)!)^{2}}\Vert v \Vert_{\lambda,q+1}\sum_{a=0}^{+\infty}
\frac{\Vert  w\Vert_{\lambda,a}}{(a!)^{2}}
\frac{(a!)^{2}(r!)^{2}((q+1)!)^{2}C^{r}_{a+r+q}C^{q}_{a+q}}{(q+1)^{2}((a+r+q)!)^{2}},
\end{equation*}
and by observing that
\begin{equation*}
\frac{(a!)^{2}(r!)^{2}((q+1)!)^{2}C^{r}_{a+r+q}C^{q}_{a+q}}{(q+1)^{2}((a+r+q)!)^{2}}=\frac{a!r!q!}{(a+r+q)!} \leq 1.
\end{equation*}
\end{proof}


{\bf Acknowledgements:}  We would like to thank the anonymous referees for  carefully reading a previous version of this work, for  helpful suggestions  in order to improve its presentation and for pointing out  a gap in one of the proofs.
\bibliographystyle{plain}

\end{document}